\documentclass[preprint,review,12pt]{elsarticle}

\usepackage{amssymb}
\usepackage{amsmath}

\usepackage{xcolor}
\usepackage[T1]{fontenc}

\journal{Theoretical Computer Science}

\newcommand{\com}{\mathbin{;}}

\newcommand{\CMTT}{\textit{MTT}}
\newcommand{\CMTTR}{\textit{MTT}^{\textit{R}}}
\newcommand{\CMTTU}{\textit{MTT}^{\textit{U}}}
\newcommand{\CATT}{\textit{ATT}}
\newcommand{\CATTR}{\textit{ATT}^{\textit{R}}}
\newcommand{\CATTU}{\textit{ATT}^{\textit{U}}}
\newcommand{\CBREL}{\textit{B}\textit{-REL}}
\newcommand{\CTREL}{\textit{T}\textit{-REL}}
\newcommand{\CTRREL}{\textit{T}^{\textit{R}}\textit{-REL}}

\newcommand{\MTT}{\text{MTT}}
\newcommand{\MTTR}{\text{MTT}^{\text{R}}}
\newcommand{\MTTU}{\text{MTT}^{\text{U}}}
\newcommand{\ATT}{\text{ATT}}

\newcommand{\ATTU}{\text{ATT}^{\text{U}}}
\newcommand{\BREL}{\textrm{B-REL}}
\newcommand{\TREL}{\text{T}\text{-REL}}
\newcommand{\TRREL}{\text{T}^{\text{R}}\text{-REL}}

\newcommand{\CMTTRC}{\textit{MTT}^{\textit{R}}_{\text{C}}}
\newcommand{\CMTTUC}{\textit{MTT}^{\textit{U}}_{\text{C}}}
\newcommand{\CMTTC}{\textit{MTT}_{\text{C}}}
\newcommand{\CMTTUFV}{\textit{MTT}^{\textit{U}}_{\text{FV}}}
\newcommand{\CMTTFV}{\textit{MTT}_{\text{FV}}}
\newcommand{\CMTTRFV}{\textit{MTT}^{\textit{R}}_{\text{FV}}}
\newcommand{\CMTTDFV}{\textit{MTT}_{\text{DFV}}}
\newcommand{\CMTTRDFV}{\textit{MTT}^{\textit{R}}_{\text{DFV}}}
\newcommand{\CMTTUDFV}{\textit{MTT}^{\textit{U}}_{\text{DFV}}}

\newcommand{\MTTDFV}{\text{MTT}_{\text{DFV}}}
\newcommand{\MTTRDFV}{\text{MTT}^{\text{R}}_{\text{DFV}}}
\newcommand{\MTTUDFV}{\text{MTT}^{\text{U}}_{\text{DFV}}}

\newcommand{\la}{\langle}
\newcommand{\ra}{\rangle}

\newcommand{\rank}{\mathrm{rk}}
\newcommand{\rhs}{\mathrm{rhs}}
\newcommand{\size}{\mathrm{size}}
\newcommand{\nf}{\mathit{nf}}
\newcommand{\ct}{\mathit{ct}}
\newcommand{\Sts}{\mathit{Sts}}

\newcommand{\Top}{\mathit{top}}

\newcommand{\dom}{\mathrm{dom}}

\def \<#1>{{\langle {#1} \rangle}}

\newcommand{\oc}{\mathrm{oc}}

\newenvironment{romanenumerate}{\begin{enumerate}}{\end{enumerate}}

\makeatother
\newcommand{\BOX}{%
  \ifmmode\else\leavevmode\unskip\penalty9999\hbox{}\nobreak\hfill\fi
  \quad\hbox{$\Box$}}
\newcommand\QED{\BOX}

\newtheorem{corollary}{Corollary}
\newtheorem{theorem}{Theorem}
\newtheorem{lemma}{Lemma}
\newtheorem{claim}{Claim}
\newdefinition{definition}{Definition}
\newproof{proof}{Proof}

\begin{document}

\begin{frontmatter}



\title{Characterizing Attributed Tree Translations in Terms of Macro Tree Transducers}


\author[label1]{Kenji Hashimoto}

\affiliation[label1]{organization={Graduate School of Informatics, Nagoya University},
            addressline={Furo-cho Chikusa-ku}, 
            city={Nagoya},
            postcode={464-8603}, 
            state={Aichi},
            country={Japan}}

\author[label2]{Sebastian Maneth}

\affiliation[label2]{organization={AG Datenbanken, Universit\"at Bremen}, 
            addressline={P.O. Box 330 440}, 
            postcode={28334}, 
            state={Bremen},
            country={Germany}}

\begin{abstract}
It is well known that attributed tree transducers can be equipped with 
``regular look-around'' in order to obtain a more robust class of
translations. We present two characterizations of this class in terms
of macro tree transducers (MTTs): the first one is a static restriction on 
the rules of the MTTs, where the MTTs need to be equipped with regular
look-around. The second characterization is a dynamic one, where
the MTTs only need regular look-ahead. 
\end{abstract}




\end{frontmatter}




\section{Introduction}

Attributed tree transducers (ATTs)~\cite{DBLP:journals/actaC/Fulop81} are
instances of attribute grammars~\cite{DBLP:journals/mst/Knuth68,DBLP:journals/mst/Knuth71} 
where the only semantic domain
is trees, and the only available operation is tree top-concatenation. 
Most definitions of ATTs define them as total deterministic devices, thus realizing
total functions from input trees to output trees (cf., e.g.,\cite{DBLP:series/eatcs/FulopV98}).
ATTs strictly generalize top-down tree transducers: an ATT without 
inherited attributes is exactly a top-down tree transducer. 
On the other hand, there are bottom-up tree transformations that \emph{cannot}
be realized by ATTs~\cite{DBLP:series/eatcs/FulopV98}. The latter ``deficiency'' can be remedied
if the ATTs are equipped with \emph{regular look-around}; 
this yields a robust class of translations which coincides 
with the class of 
tree-to-dag (to tree) translations definable in MSO logic~\cite{BE00}.
Intuitively, regular look-around amounts to preprocessing the input tree
via an “attributed relabeling”; such a relabeling is equivalent to a
top-down relabeling with regular look-ahead (or, to preprocessing the
input tree first via a bottom-up relabeling, followed by a top-down relabeling)~\cite{EM99}. 

While ATTs can be seen as an operational model, 
macro tree transducers (MTTs)~\cite{DBLP:journals/jcss/EngelfrietV85} 
(introduced independently in~\cite{Eng80} and 
in~\cite{DBLP:journals/tcs/CourcelleF82,DBLP:journals/tcs/CourcelleF82a})
are a denotational model, more akin to functional programming, see e.g.~\cite{DBLP:conf/icfp/BahrD13}. 
MTTs are strictly more powerful than ATTs.
For instance, they can translate a monadic input tree of height $n$
into a monadic output tree of height $2^n$.
An ATT cannot do this, because the number of distinct subtrees of an output tree is linearly
bounded in
the size of the corresponding input tree (property ``\textbf{LIN}'';
see, e.g.,~\cite{DBLP:series/eatcs/FulopV98}).
In the context of XML translation, ATTs with nested pebbles have been
considered~\cite{DBLP:journals/jcss/MiloSV03} which can be simulated by compositions
of MTTs~\cite{DBLP:journals/acta/EngelfrietM03}.

F{\"u}l{\"o}p and Vogler~\cite{FV99} were the first to present a characterization
of ATTs in terms of MTTs. Their characterization is based on a particular 
construction from ATTs to MTTs: each synthesized attribute of the ATT becomes
a state of the MTT, and each inherited attribute becomes a parameter of a state
of the MTT. 
This causes that for a given input symbol the parameter trees of all
state calls in the rules for that symbol are equal.
This equality is the first part of their restriction (called \emph{consistency}) on MTTs. 
This may be a somewhat inflexible syntactic restriction. In fact the non-deleting normal form 
of some consistent MTT is not consistent.
In their construction from MTT to ATT, it may happen that the resulting
ATT is \emph{circular} (and hence not well-defined). 
Thus, the second part of their restriction requires that the ATT resulting from their
construction is \emph{non-circular}. 
Let us refer to their (full) restriction as \textit{FV-orig}.

In this paper we formulate restrictions on MTTs to characterize ATTs.
Our first restriction, called \emph{FV property} is similar to 
consistency (part one of FV-orig) and differs in the following ways: (1) it is defined for MTTs with regular look-around, 
(2) the MTTs are nondeleting (with respect to the parameters), and (3) for such MTTs we can construct ATTs 
with regular look-around that are guaranteed to be non-circular. 
In the FV property we relax  the requirement of the consistency on equality between the parameter trees of all state calls in the rules 
for the same symbols by using a ``parameter renaming mapping''.

Our second characterization of ATTs with regular look-around in terms of MTTs
(with regular look-ahead) is called 
the \emph{dynamic FV property}.
It generalizes the FV property in several ways.
First, we do not compare argument trees of different states.
Second, argument trees (of the same state) are not required to be syntactically equal,
but are only required to be semantically equal (i.e., must evaluate to the same output tree).
And third, we only consider states and argument trees that are reachable. 
In this characterization, it suffices to consider MTTs with regular look-ahead.
The reason for this is that MTTs with the dynamic FV property are closed under
left-composition with top-down relabelings; a property that we conjecture not to
hold for MTTs with the FV property and regular-look ahead only. 
%
It should be noted that the generality of the dynamic FV property comes at a price:
while the FV property is easily decidable, we do not know how to decide the dynamic FV property.
In fact, we are able to show that deciding the dynamic FV property is \emph{at least} as
difficult as deciding equivalence of ATTs. 
The equivalence problem of ATTs is a long standing open problem;
for some subclasses of ATTs with regular look-around, e.g., MSO-definable tree-to-tree translations,
equivalence is decidable~\cite{Eng80,DBLP:journals/jacm/SeidlMK18}.
Note that MTTs with monadic output trees (i.e., a class incomparable to ATTs) has
decidable equivalence~\cite{DBLP:journals/jacm/SeidlMK18}
(see also, \cite{DBLP:conf/fsttcs/BoiretPS18,DBLP:journals/siglog/Bojanczyk19a}).
---
All our results also hold for \emph{partial} transducers; this can be achieved by allowing
the look-ahead and look-around to be partial. The characterization by F{\"u}l{\"o}p and Vogler~\cite{FV99}
is defined for total transducers only.

\section{Preliminaries}


For a non-negative integer $k$, we denote by $[k]$ the
set $\{1,\dots,k\}$. 
For functions $f,g$ we denote by $f\com g$ their
sequential composition that maps $a\in\dom(f)$ to $g(f(a))$.
For classes of functions $F$ and $G$, 
$F\com G$ denotes $\{ f\com g\mid f\in F, g\in G\}$.


\medskip

\noindent
\textbf{Trees and Tree Substitution.}\quad
A ranked alphabet $\Sigma$ consists of an alphabet together with
a mapping $\rank_\Sigma$ that associates to each symbol of the alphabet
a non-negative integer (its \emph{rank}). 
By $\Sigma^{(k)}$ we denote the set of symbols $\sigma$ in $\Sigma$
for which $\rank_\Sigma(\sigma)=k$. 
We also write $\sigma^{(k)}$ to indicate that the symbol
$\sigma$ has rank $k$. 
The set $T_\Sigma$ of \emph{trees} (\emph{over $\Sigma$}) is 
the smallest set $T$ of strings such that 
if $t_1,\dots,t_k\in T$, $k\geq 0$, and $\sigma\in\Sigma^{(k)}$,
then $\sigma(t_1,\dots,t_k)$ is in $T$.
We write $\sigma$ for a tree of the form $\sigma()$.
For a set $S$ disjoint with $\Sigma$ we define $T_\Sigma(S)$ as
$T_{\Sigma'}$ where $\Sigma'=\Sigma\cup\{ s^{(0)}\mid s\in S\}$.
The set $V(t)$ of nodes of a tree $t=\sigma(t_1,\dots,t_k)$ is defined as
$\{\varepsilon\}\cup \{ i u\mid i\in[k], u\in V(t_i) \}$. 
For $u, u'\in V(t)$, we say that $u$ is a descendant of $u'$ if $u'$ is a prefix of $u$. 
If $u$ is a descendant of $u'$ and $u\neq u'$, we say that $u$ is a proper descendant of $u'$. 
We let $u0=u$ for every node $u$.
For $u\in V(t)$, $t[u]$ denotes the label of $u$,
$t/u$ denotes the subtree rooted at $u$,
and for a tree $t'$, $t[u\leftarrow t']$ denotes the tree obtained from $t$ by
replacing the subtree rooted at $u$ by the tree $t'$.

Let $t,t_1,\dots,t_k$ be trees and $\sigma_1,\dots,\sigma_k$ be symbols of rank zero.
Then $t[\sigma_i\leftarrow t_i\mid i\in[k]]$ denotes the result of replacing
in $t$ each occurrence of $\sigma_i$ by the tree $t_i$.
We also define a more powerful type of tree substitution, where inner nodes
of a tree may be replaced.
We fix the set $Y=\{ y_1,y_2,\dots \}$ of \emph{parameters}, and denote $\{y_1,\dots,y_m\}$ by $Y_m$.
Now let $\sigma_1,\dots,\sigma_k$ be symbols of arbitrary rank $m_1,\ldots,m_k$, respectively, and $t_i\in T_\Sigma(Y_{m_i})$ for $i\in [k]$. 
Then $t[\![ \sigma_i\leftarrow t_i\mid i\in[k] ]\!]$ denotes the result
of replacing in $t$ each subtree $\sigma_i(s_1,\dots,s_{m_i})$
by the tree $t_i[y_j\leftarrow s_j'\mid j\in[m_i]]$ where 
$s_j'=s_j[\![ \sigma_i\leftarrow t_i\mid i\in[k] ]\!]$.


\medskip

\noindent
\textbf{Macro Tree Transducers.}\quad
We fix the set $X=\{x_1,x_2, \dots\}$ of \emph{input variables} and
assume it to be disjoint from $Y$ and all other alphabets. 
We denote $\{x_1,\dots,x_k\}$ by $X_k$. 
A \emph{(total deterministic) macro tree transducer} (\emph{MTT} for short) is
a tuple $M=(Q,\Sigma,\Delta,q_0,R)$ where $Q$ is a ranked alphabet of states,
$\Sigma$ and $\Delta$ are ranked alphabets
of input and output symbols, respectively, $q_0\in Q^{(0)}$ is
the initial state, and $R$ is a set of rules.
For every $q\in Q^{(m)}$ and $\sigma\in\Sigma^{(k)}$ there is
exactly one rule in $R$ of the form
$\langle q,\sigma(x_1,\dots,x_k)\rangle(y_1,\dots,y_m)\to \zeta$
where $\zeta\in T_{\Delta\cup \langle Q,X_k\rangle}(Y_m)$. 
Note that we follow the style of the definition of MTTs in \cite{EM99}. 
For a set $S$, 
$\langle Q,S\rangle$ denotes the ranked set 
$\{\langle q,s\rangle^{(n)}\mid q\in Q^{(n)}, n\geq 0, s\in S\}$. 
Note that we use this notation even for an infinite set $S$, 
e.g., $\langle Q, X\rangle$ for the variable set $X$.  
The tree $\zeta$ in the right-hand side is also denoted $\rhs_M(q,\sigma)$.
The semantics of an MTT is defined as follows.
Let $\sigma\in\Sigma^{(k)}$, $k\geq 0$,
$s_1,\dots,s_k \in T_\Sigma$, and $q\in Q^{(m)}$ with $m\geq 0$.
Then $M_q(\sigma(s_1,\dots,s_k))$ denotes the tree
$\rhs_M(q,\sigma)[\![\langle q',x_i\rangle \leftarrow M_{q'}(s_i)\mid
q'\in Q, i\in[k]]\!]$.
The \emph{translation realized by $M$}, denoted $\tau_M$,
is defined as $\{ (s,t)\mid s\in T_\Sigma, t=M_{q_0}(s)\}$.
The class of all translations realized by MTTs is denoted $\CMTT$.


\medskip
\noindent
\textbf{Attributed Tree Transducers.}\quad
An \emph{attributed tree transducer} is a tuple
$A=(S,I,\Sigma,\Delta,$ $\alpha_0,R)$ where $S$ and $I$ are disjoint finite sets
of synthesized and inherited attributes, respectively, $\Sigma$ and $\Delta$ 
are ranked alphabets of input and output symbols, respectively, 
$\alpha_0\in S$ is the output attribute, and $R$ is a collection of
sets $R_\sigma$, $\sigma\in\Sigma$ of rules. Let $\sigma\in\Sigma^{(k)}$ and $k\geq 0$. 
Let $\pi$ be a variable for paths, and we define $\pi 0 = \pi$.  
For every $\alpha\in S$ the set $R_\sigma$ contains a rule
of the form $\alpha(\pi)\to t$ and for every $\beta\in I$ and
$i\in[k]$ the set $R_\sigma$ contains a rule of the form
$\beta(\pi i)\to t'$, and the trees $t,t'$ are in $T_\Delta( \{\alpha'(\pi i)\mid \alpha'\in S, i\in [k]\}
\cup \{ \beta'(\pi)\mid \beta'\in I\})$. 
The right-hand sides $t,t'$ of these rules are
denoted $\rhs_A(\sigma,\alpha(\pi))$ and 
$\rhs_A(\sigma,\beta(\pi i))$, respectively.

To define the semantics of the ATT $A$ on an input tree $s\in T_\Sigma$, 
we first define the \emph{dependency graph of $A$ on $s$} as
$D_A(s)=(V,E)$, where 
$V=\{(\alpha_0,\varepsilon)\}\cup (S\cup I)\times (V(s)-\{\varepsilon\})$ and
$E=\{ ((\gamma',uj),(\gamma,ui))\mid u\in V(s), \gamma'(\pi j)\text{ occurs in } \rhs_A(s[u],\gamma(\pi i)), 0\leq i,j\leq \rank_\Sigma(s[u]), \gamma, \gamma'\in S\cup I\}$.
If $D_A(s)$ contains a cycle for some $s\in T_\Sigma$ then
$A$ is called \emph{circular}.
Let $N=\{\alpha_0(\varepsilon)\}\cup \{a(u)\mid a\in S\cup I, u\in V(s)-\{\varepsilon\}\}$. 
For trees $t,t'\in T_\Delta(N)$, $t\Rightarrow_{A,s} t'$ holds if
$t'$ is obtained from $t$ by 
replacing a node labeled $\gamma(ui)$ by $\rhs_A(s[u],\gamma(\pi i))[\gamma'(\pi i)\leftarrow \gamma'(u i) \mid \gamma'\in S\cup I, 0\leq i\leq \rank_\Sigma(s[u])]$. 
If $A$ is non-circular, then every $t\in T_\Delta(N)$ has a unique normal form
with respect to $\Rightarrow_{A,s}$ which we denote by $\nf(\Rightarrow_{A,s},t)$.
The \emph{translation realized by $A$}, denoted $\tau_A$, is defined as
$\{(s,\nf(\Rightarrow_{A,s},\alpha_0(\varepsilon)))\mid s\in T_\Sigma\}$.
The class of all translations realized by non-circular ATTs is denoted $\CATT$.


\medskip
\noindent
\textbf{Regular Look-Ahead and Regular Look-Around.}~
For the classes $\CMTT$ and $\CATT$
we define \emph{regular look-ahead} and \emph{regular look-around} by
means of pre-composition with the classes $\BREL$ and $\TRREL$, respectively.
Here $\BREL$ is the class of deterministic bottom-up finite state 
relabelings $(P,\Sigma,\Sigma',F,R)$ where $P$ is a finite set of states, 
$\Sigma, \Sigma'$ are ranked alphabets, the set $F\subseteq P$ of final states, and 
a set $R$ of relabeling rules. The set $R$ contains for every
$p_1,\dots,p_k\in P$ and $\sigma\in\Sigma^{(k)}$, $k\geq 0$, exactly one rule of the form
$\sigma(p_1(x_1),\dots,p_k(x_k))\to p(\sigma'(x_1,\dots,x_k))$
where $p\in P$ and $\sigma'\in\Sigma'^{(k)}$.
They are evaluated in the obvious way, like a bottom-up tree automaton with the set $F$ of final states.
Note that we assume $F=P$ in this paper since we treat only total transducers for simplicity of discussion, 
but our results can be extended to \emph{partial} transducers by allowing relabeling transducers to be partial. 
The class $\TRREL$ is defined as $\BREL\com \TREL$, where
$\TREL$ is a deterministic top-down relabeling 
$(Q,\Sigma,\Sigma',q_0,R)$ which contains for every $q\in Q$ and $\sigma\in\Sigma^{(k)}$
exactly one rule of the form $q(\sigma(x_1,\dots,x_k))\to\sigma'(q_1(x_1),\dots,q_k(x_k))$,
where $\sigma'\in\Sigma'^{(k)}$ and $q_1,\dots,q_k\in Q$.
We use the superscripts \text{R} (\text{U}) to indicate the presence of regular 
look-ahead (look-around).
For instance $\CATTU = \CTRREL\com \CATT$ is the class of translations realized by attributed tree
transducers with regular look-around (cf. Section~4 of~\cite{EM99}); this is the class
for which we give new characterizations in terms of $\MTT$s in this paper.

\section{FV Property for Nondeleting MTTs with Regular Look-Around}

Let us recall the definition of the consistency of F{\"u}l{\"o}p and Vogler~\cite{FV99}. 
In the following, let $M=(Q,\Sigma,\Delta,q_0,R)$ be an $\MTT$. 

\begin{definition}
\label{def:importance}
\rm
Let $\sigma\in \Sigma^{(k)}$, $k\geq 0$, and $q\in Q$. 
Let $\zeta=\rhs_M(q, \sigma)$ and $v\in V(\zeta)$.  
The node $v$ is \emph{important} in $\zeta$ for $s_1,\ldots,s_k\in T_\Sigma$
if the symbol $*$ occurs in $\zeta[v\leftarrow *][\![\<q', x_i> \leftarrow M_{q'}(s_i)\mid q'\in Q, i\in [k]]\!]$. 
If $v$ is important in $\zeta$ for some  $s_1,\ldots,s_k\in T_\Sigma$, then
we say that $v$ is \emph{important} in $\zeta$. 
\end{definition}

\begin{definition}\label{def:consistent}
\rm
The $\MTT$ $M$ 
is \emph{consistent} if the following condition holds 
for every $\sigma\in \Sigma^{(k)}$, $k\geq 0$, and $q_1,q_2\in Q$: 
for $i=1,2$ let $\zeta_i=\rhs_M(q_i,\sigma)$ 
and $w_i\in V(\zeta_i)$ 
such that $\zeta_i[w_i]\in \<Q, \{x_l\}>$ for some $l\in [k]$. 
Let $\<p_i^{(m_i)}, x_l>=\zeta_i[w_i]$ for $i=1,2$. 
Then, for every $j\in [\min\{m_1, m_2\}]$, if $w_1j$ and $w_2j$ are important in $\zeta_1$ and $\zeta_2$, respectively, 
then $\Top(\zeta_1/w_1j)=\Top(\zeta_2/w_2j)$ where $\Top(\zeta)$ is defined as follows:
\[
\Top(\zeta) = 
\begin{cases}
q'(\pi i)	 & \text{if $\zeta[\varepsilon]=\langle q', x_i\rangle \in \<Q,X>$}\\
\delta(\Top(\zeta_1),\ldots,\Top(\zeta_k)) & \text{if $\zeta=\delta(\zeta_1,\ldots,\zeta_k)$ where $\delta\in \Delta$}\\
y_j(\pi) & \text{if $\zeta=y_j\in Y$.}
\end{cases}
\]
\end{definition}
The class of all translations realized by consistent MTTs is denoted $\CMTTC$. 

In Figure~\ref{fig:abcd} on the left we see the dependency graph of an ATT which translates monadic trees
of the form $\#(a^n(e))$ to monadic output trees of the form $a^n(b^n(c^n(d^n(e))))$; here shown for $n=2$. 
The top right part shows the rules of a consistent MTT. Observe the two bottom-most rules, where the state 
$q_1$ deletes its second parameter ($y_2$) and the state $q_2$ deletes its first parameter, 
and the two rules above which have $\bot$ at the non-important argument positions in their right-hand sides. 
Note that non-important parameters 
just for padding are necessary for an MTT to satisfy the consistency. The bottom right part in
Figure~\ref{fig:abcd} 
shows an equivalent MTT without the redundant parameters which does not satisfy the consistency. 

\begin{figure}[t]
\centering
\begin{minipage}{0.3\textwidth}
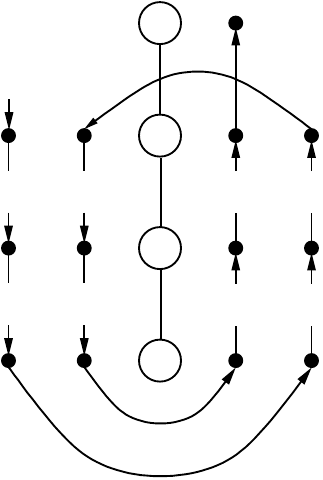
\end{minipage}
\begin{minipage}{0.65\textwidth}
\[
\begin{array}{lcl}
\langle q_0,\#(x_1)\rangle&\to&\langle q_1,x_1\rangle(\langle q_2,x_1\rangle(\bot,e),\bot)\\
\langle q_1,a(x_1)\rangle(y_1,y_2)&\to& a(\langle q_1,x_1\rangle(b(y_1),\bot))\\
\langle q_2,a(x_1)\rangle(y_1,y_2)&\to& c(\langle q_2,x_1\rangle(\bot,d(y_2)))\\
\langle q_1,e\rangle(y_1,y_2) &\to& y_1\\
\langle q_2,e\rangle(y_1,y_2) &\to& y_2\\[1mm]
\hline\\[-4mm]
\langle q_0,\#(x_1)\rangle&\to&\langle q_1,x_1\rangle(\langle q_2,x_1\rangle(e))\\
\langle q_1,a(x_1)\rangle(y_1)&\to& a(\langle q_1,x_1\rangle(b(y_1)))\\
\langle q_2,a(x_1)\rangle(y_1)&\to& c(\langle q_2,x_1\rangle(d(y_1)))\\
\langle q_1,e\rangle(y_1) &\to& y_1\\
\langle q_2,e\rangle(y_1) &\to& y_1
\end{array}
\]
\end{minipage}
\caption{An ATT and MTTs that translate $\#(a^n(e))$ to $a^n(b^n(c^n(d^n(e))))$}
\label{fig:abcd}
\end{figure}

F{\"u}l{\"o}p and Vogler construct from a consistent MTT an ATT which
may possibly be circular. Let us review such an example from~\cite{FV99}.
Consider the MTT with the rules shown in the left of Figure~\ref{fig:circ}, where
all rules with left-hand sides that are \emph{not} shown, have
as right-hand side the output leaf $\#$.
\begin{figure}[t]
\begin{minipage}{0.68\textwidth}
\[
\begin{array}{lcl}
\la q_0,\sigma(x_1,x_2)\ra&\to&\la q,x_1\ra(\la q_2,x_2\ra(\la q_1,x_1\ra(\$)))\\
\la q,e\ra(y_1)&\to&y_1\\
\la q_2,e\ra(y_1)&\to&y_1\\
\la q_1,e'\ra(y_1)&\to&y_1\\[1cm]
\end{array}
\]
\end{minipage}
\hspace*{-3.0cm}
\begin{minipage}{0.40\textwidth}
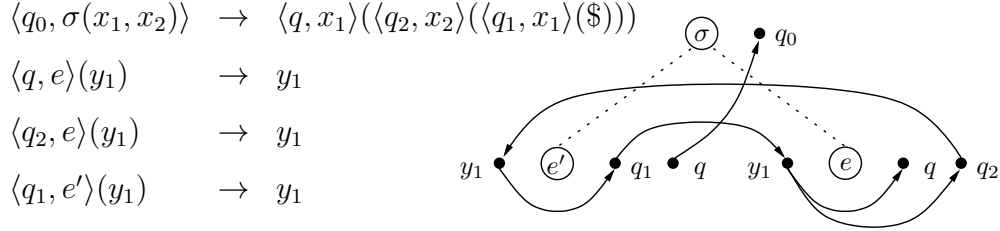
\end{minipage}
\caption{The consistent MTT $M_c$ and the circular ATT obtained by the construction of~\cite{FV99}}
\label{fig:circ}
\end{figure}
As the reader may verify, this MTT indeed is consistent:
the $(q_0,\sigma)$-rule is consistent, because the parameter of $q$ is important
(for inputs of the form $\sigma(e,s)$), 
but the one of $q_1$ is not important (because it is 
deleted by $q_1$ for inputs $\sigma(e,s)$, or by $q$ for inputs $\sigma(e',s)$). 
A part of the dependency graph for the input tree $\sigma(e',e)$
is shown in Figure~\ref{fig:circ}.
The idea is as follows. Since $\la q,x_1\ra$ appears in a non-nested position
in $\rhs(q_0,\sigma)$, there is an edge from $(q,1)$ to $(q_0,\varepsilon)$.
The inherited attribute ``$y_1$'' at the $e'$-node has an incoming edge
from $(q_2,2)$, because the call $\la q_2,x_2\ra$ appears in the first parameter
position of $\langle q,x_1\ra$ in $\rhs(q_0,\sigma)$. The bottom-most
edges from the nodes $(y_1,1)$ and $(y_1,2)$ stem from rules
with right-hand side $y_1$.


We now define the FV property for nondeleting MTTs, by adapting the definition of consistency to nondeleting MTTs.
An MTT is \emph{nondeleting}, if every parameter that appears in the left-hand side
of a rule, also appears in the right-hand side. 
In order to formalize MTTs such as $M_{a}$ in the bottom-right of Figure~\ref{fig:abcd}
we use mappings of the form $\rho: Q\times\mathbb{N}\to\mathbb{N}$. Intuitively,
if $\rho(q,j)=j'$, then the $j$th parameter of state $q$ of an MTT is renamed by the parameter with index $j'$ in the corresponding consistent MTT. 
So, for $M_a$ it holds that $\rho(q_1,1)=1$ and $\rho(q_2,1)=2$. 
As we will state, in our construction of an equivalent ATT via a consistent MTT, the parameter with the index $j'$ of the state $q$ 
corresponds to the $j'$th inherited attribute of the resulting ATT.

\begin{definition}\label{def:equality-after-renaming}
\rm
Let $\rho: Q\times \mathbb{N}\to \mathbb{N}$. 
For $q\in Q^{(m)}$, we denote by $\Psi^\rho_q$ the renaming of parameters $[y_l\leftarrow y_{l'} \mid l\in [m], l'=\rho(q, l)]$. 
For $q,q'\in Q$, 
we define the relation $\sim^{q,q'}_\rho$ on 
trees $\xi_1,\xi_2\in T_{\Delta\cup\<Q,X>}$ such that
$\xi_i[\varepsilon]\in\<Q,X>$ as follows:
$\xi_1 \sim^{q,q'}_\rho \xi_2$ if $\xi_1[\varepsilon]=\<q_1^{(m_1)}, x>$ and $\xi_2[\varepsilon]=\<q_2^{(m_2)}, x>$ for some $q_1,q_2\in Q$ and 
$x\in X$, and for every $j_1\in [m_1]$ and $j_2\in [m_2]$, if  $\rho(q_1, j_1)=\rho(q_2,j_2)$ then $(\xi_1/j_1)\Psi^{\rho}_{q} = (\xi_2/j_2)\Psi^{\rho}_{q'}$. 
\end{definition}

\begin{definition}
\rm
\label{def:FVproperty}
Let $\rho:Q \times \mathbb{N} \to \mathbb{N}$
such that for every $q\in Q$, 
$\rho(q,i)\not=\rho(q,j)$ if $i\not=j$ and $i,j\in[\rank_Q(q)]$. 
Suppose that the $\MTT$ $M$ is nondeleting. 
The $\MTT$ $M$ 
\emph{has the FV property with $\rho$} 
if the following condition holds for every $\sigma\in \Sigma^{(k)}$, $k\geq 0$, and $q_1,q_2\in Q$:
let $\xi_1$ and $\xi_2$ be any subtrees of $\rhs_M(q_i, \sigma)$ for $i=1,2$, respectively, 
such that $\xi_1[\varepsilon],\xi_2[\varepsilon]\in \<Q,\{x_j\}>$ for some $j\in [k]$. 
Then $\xi_1\sim^{q_1,q_2}_{\rho} \xi_2$ must hold.
We say that $M$ \emph{has the FV property} if it has the FV property with some $\rho$. 
We denote by $\CMTTFV$ the class of all translations realized by $\MTT$s with the FV property. 
\end{definition}
%


We now present our construction from MTTs with the FV property to consistent MTTs. 
Intuitively, according to the mapping $\rho$, in the construction, 
we increase the number of parameters of every state to the maximum number $\tilde{m}$ of the range of $\rho$, 
and each $j$th argument tree of each call $\<q, x_i>$ in the right-hand side of rules are moved 
to the $\rho(q, j)$th argument position of the corresponding call $\<q, x_i>^{(\tilde{m})}$. 
The remaining argument parts of state calls are filled by a dummy label $\bot$. 

\newcommand{\expand}{e}

\begin{definition}
\rm
\label{def:MTTFV2MTTC}
Suppose that the nondeleting MTT $M$ satisfies the FV property with $\rho:Q\times\mathbb{N}\to \mathbb{N}$.
Let $\tilde{m}=\max\{\rho(q,j) \mid q\in Q, j\in [\rank_Q(q)]\}$. 
The \emph{consistent MTT associated with $M$}, denoted by $\mathcal{E}(M)$, is the MTT $(Q_e,\Sigma,\Delta,q_0,R_e)$ with $Q_e=Q^{(0)}\cup \{q^{(\tilde{m})}\mid q\in Q-Q^{(0)}\}$ 
and $R_e$ is the smallest set $R'$ satisfying the following. 
Note that each state of $\mathcal{E}(M)$ corresponds to the state with the same name in $M$ 
but its rank is $\tilde{m}$ if its original rank is not zero, and zero otherwise. 
For every $q\in Q_e$ and $\sigma \in \Sigma^{(k)}$ with $k\geq 0$, 
$R'$ contains the $(q, \sigma)$-rule with $\expand_q(\rhs_M(q,\sigma))$ in the right-hand side
where $\expand_q(\zeta)$ is defined as follows:
\begin{itemize}
    \item $\expand_q(\zeta)=\<q',x_i>$ if $\zeta=\<q',x_i>\in \<Q^{(0)}, X_k>$. 
    \item  $\expand_q(\zeta)=\<q', x_i>(s_1,\ldots,s_{\tilde{m}})$ 
        if $\zeta=\langle q', x_i\rangle(\zeta_1,\ldots,\zeta_m)$ such that $\<q', x_i>\in \<Q^{(m)}, X_k>$ and $m>0$, 
        where for $j\in [\tilde{m}]$, $s_j=\expand_q(\zeta_i)$ 
        if there exists $i$ such that $\rho(q', i)=j$, and  
        $s_j=\bot$ otherwise, where $\bot$ is any label in $\Delta^{(0)}$. 
    \item  $\expand_q(\zeta)=\delta(\expand_q(\zeta_1),\ldots,\expand_q(\zeta_{m}))$ 
        if $\zeta=\delta(\zeta_1,\ldots,\zeta_m)$ with $\delta\in \Delta$.
    \item $\expand_q(\zeta)=y_{\rho(q,j)}$ if $\zeta=y_j\in Y_{\rank_Q(q)}$.
\end{itemize}
\end{definition}

We can see that the above construction transforms the bottom right MTT 
with the FV property in Figure~\ref{fig:abcd} into the top right MTT, 
which is consistent. 

We shall prove the correctness of the construction of Definition~\ref{def:MTTFV2MTTC}. 
We show that $\mathcal{E}(M)$ is a consistent MTT equivalent with $M$ 
if the nondeleting MTT $M$ has the FV property. 

\begin{lemma}
\label{lem:equivalence_MTTFV2MTTC}
Suppose that the nondeleting MTT $M$ has the FV property with $\rho$. Then $M$ and $\mathcal{E}(M)$ are equivalent. 
\end{lemma}
\begin{proof}
Let $M'=\mathcal{E}(M)=(Q',\Sigma,\Delta,q_0,R')$. 
Recall that $\Psi^\rho_q=[y_l\leftarrow y_{l'}\mid l\in [m], l'=\rho(q,l)]$ for $q\in Q^{(m)}$ and $m\geq 0$. 
We prove the following two statements (i) and (ii). The lemma follows from (i). 
Let $s$ be an arbitrary tree in $T_\Sigma$. 
\begin{itemize}
    \item[(i)] For every $q\in Q$, $M_q(s)\Psi^\rho_q=M'_q(s)$.
    
    \item[(ii)] For every $q\in Q$, $u\in V(\zeta_q)$, and $k\geq 0$ such that 
    $s[\varepsilon]\in \Sigma^{(k)}$, 
    $(\zeta_q/u)\theta\Psi^\rho_q=\expand_q(\zeta_q/u)\theta'$ where $\zeta_q=\rhs_M(q,\sigma)$, 
    $\theta=[\![\langle r, x_i\rangle \leftarrow M_r(s/i) \mid r\in Q, i\in [k]]\!]$, and
    $\theta'=[\![\langle r, x_i\rangle \leftarrow M'_r(s/i) \mid r\in Q', i\in [k]]\!]$. 
\end{itemize}
We first prove that (ii) implies (i) for all $s\in T_\Sigma$.

\paragraph*{(ii) implies (i)}
Let $s=\sigma(s_1,\ldots, s_k)\in T_\Sigma$ where $\sigma\in \Sigma^{(k)}$, $k\geq 0$, and 
$s_1,\ldots, s_k\in T_\Sigma$. 
Let $q\in Q^{(m)}$.  
Let $\zeta=\rhs_M(q,\sigma)$, 
$\theta=[\![\langle r, x_i\rangle \leftarrow M_r(s/i) \mid r\in Q, i\in [k]]\!]$, 
and $\theta'=[\![\langle r, x_i\rangle \leftarrow M'_r(s/i) \mid r\in Q', i\in [k]]\!]$. 
Assume that (ii) holds for $s$. With (ii) for $u=\varepsilon$, we get $M_q(s)\Psi^{\rho}_q = \zeta\theta\Psi^{\rho}_q = \expand_q(\zeta)\theta' = M'_q(s)$ because $\rhs_{M'}(q,\sigma)=\expand_q(\zeta_q)$ by the construction of Definition~\ref{def:MTTFV2MTTC}. 

Next, with the property that (ii) implies (i), we prove that (ii) holds for all $s\in T_\Sigma$ 
by induction on the structure of $s$. 
They imply that (i) also holds for all $s$. 

\paragraph*{Base case of (ii)} 
Let $s=\sigma\in \Sigma^{(0)}$ and $k=0$. Let $q\in Q^{(m)}$,  
$\zeta=\rhs_M(q,\sigma)$, and let $u\in V(\zeta)$ and $\xi=\zeta/u$. 
Note that $\zeta \in T_{\Delta}(Y_m)$ and thus $\xi\in T_{\Delta}(Y_m)$. 
By the definition of $\expand_q$ and $\xi\in T_{\Delta}(Y_m)$, 
we get $\xi\Psi^\rho_q = \expand_q(\xi)$. 
Since $\theta$ and $\theta'$ are empty substitutions when $k=0$, 
$\xi\theta\Psi^\rho_q=\expand_q(\xi)\theta'$.

\paragraph*{Induction step of (ii)}
Let $s=\sigma(s_1,\ldots, s_k)\in T_\Sigma$ where $\sigma\in \Sigma^{(k)}$, $k>0$, and 
$s_1,\ldots, s_k\in T_\Sigma$. 
Let $q\in Q^{(m)}$.  
Let $\zeta_q=\rhs_M(q,\sigma)$, $\theta=[\![\langle r, x_i\rangle \leftarrow M_r(s/i) \mid r\in Q, i\in [k]]\!]$, and $\theta'=[\![\langle r, x_i\rangle \leftarrow M'_r(s/i) \mid r\in Q', i\in [k]]\!]$.
Now we prove that $\xi\theta\Psi^{\rho}_q=\expand_q(\xi)\theta'$ holds 
for every subtree $\xi$ of $\zeta_q$ by induction on the structure of $\xi$, which implies that (ii) holds for $s$. Let us denote by (IHs) the induction hypothesis of the outer induction on $s$, and 
by (IH$\xi$) that of the inner induction on $\xi$. 
\begin{description}
\item[Case 1.] $\xi=y_j$. It is trivial from the fact that $\Psi^\rho_p(y_j)=y_{\rho(q,j)}=\expand_q(y_j)$. 

\item[Case 2.] $\xi=\delta(\xi_1,\ldots,\xi_l)$. It holds that $\xi\theta\Psi^\rho_q=\delta(\xi_1\theta\Psi^\rho_q, \ldots, \xi_l\theta\Psi^\rho_q)$ and 
$\expand_q(\xi)\theta'=\delta(\expand_q(\xi_1)\theta', \ldots, \expand_q(\xi_l)\theta')$. 
It follows by (IH$\xi$) that $\xi_i\theta\Psi^{\rho}_q = \expand_q(\xi_i)\theta'$ 
for every $i\in [l]$. 

\item[Case 3.] $\xi=\<q', x_i>$ with $q'\in Q^{(0)}$. 
By applying (IHs) to $s/i$ and using the fact that (ii) implies (i), 
$M_{q'}(s/i)\Psi^\rho_q=M'_{q'}(s/i)$ holds.
Thus, we get $\xi\theta\Psi^\rho_q=M_{q'}(s/i)\Psi^\rho_q=M'_{q'}(s/i)=\<q',x_i>\theta'=\expand_q(\xi)\theta'$. 

\item[Case 4.] $\xi=\langle q', x_i\rangle(\xi_1,\ldots,\xi_{m'})$ with $q'\in Q^{(m')}$ and $m'>0$. We get (a) $M_{q'}(s/i)\Psi^\rho_q=M'_{q'}(s/i)$ by (IHs) and (ii)$\Longrightarrow$(i). 
Recall that (b) $\expand_p(\xi)=\<q',x_i>(\xi'_1,\ldots,\xi'_{\tilde{m}})$ where for $j\in [\tilde{m}]$, 
$\xi'_j=\expand_q(\xi_l)$ if $\rho(q, l)=j$, and $\xi'_j=\bot$ otherwise. 
In addition, it follows from (a) that (c) 
the parameters appearing in $M'_{q'}(s/i)$ are in $\{y_j \mid j=\rho(q', l), l\in [m']\}$. 
Hence, 
\begin{align*}
\xi\theta\Psi^\rho_q
&= M_{q'}(s/i)[y_l\leftarrow \xi_l\theta \mid l\in [m']]\Psi^\rho_q\\
&= M_{q'}(s/i)[y_l\leftarrow \xi_l\theta\Psi^\rho_q \mid l\in [m']] \\
&= M_{q'}(s/i)[y_l \leftarrow \expand_p(\xi_l)\theta' \mid l\in [m']]& \text{(by (IH$\xi$))}\\
&= M_{q'}(s/i)\Psi^\rho_q[y_j\leftarrow \expand_p(\xi_l)\theta' \mid j=\rho(q',l), l\in [m']]\\
&= M'_{q'}(s/i)[y_j\leftarrow \expand_p(\xi_l)\theta' \mid j=\rho(q',l), l\in [m']] & \text{(by (a))}\\
&= M'_{q'}(s/i)[y_j\leftarrow \xi'_j\theta' \mid j\in [\tilde{m}]] & \text{(by (b),(c))}\\
&= \<q',x_i>(\xi'_1,\ldots,\xi'_{\tilde{m}})\theta'=\expand_q(\xi)\theta'.
\end{align*}
\end{description}
\QED
\end{proof}

The resulting consistent MTT $\mathcal{E}(M)$ has additional properties. 
We say that the $j$th parameter of state $q$ in $M$ is \emph{permanent} 
if for every $q'\in Q$, $\sigma\in \Sigma$, $v\in V(\zeta)$ where $\zeta=\rhs_M(q', \sigma)$, 
if $\zeta[v]=\<q,x_i>$ for some $x_i\in X$ then $vj$ is important in $\zeta$. 

\begin{lemma}
\label{lem:FV2C_PROP}
Suppose that the nondeleting MTT $M$ has the FV property with $\rho$. 
Let $M'=\mathcal{E}(M)$. Then $M'$ has the following properties. 
\begin{itemize}
    \item[(i)] In the right-hand side $\zeta$ of any rule of $M'$, 
    for every $v\in V(\zeta)$, 
    if $v=v'j$ with some $v'\in V(\zeta)$ and $j\in [\tilde{m}]$ such that 
    $\zeta[v]=\<q', x_i>\in \<Q_e,X>$ and $\rho(q', l)\neq j$ for all $l\in [\rank_Q(q')]$, then $v$ is not important in $\zeta$ and $\zeta[v]=\bot$; otherwise, $v$ is important in $\zeta$.    
    \item[(ii)] The $j$th parameter of $q$ is permanent in $M'$ if $y_j\in Y$ appears in $\rhs_{M'}(q, \sigma)$ for some $\sigma\in \Sigma$. 
\end{itemize}
\end{lemma}
\begin{proof}
Let $\zeta=\rhs_{M'}(q, \sigma)$.  
We can prove Property~(i) by induction on the distance from the root in $\zeta$. 
The root is trivially important in $\zeta$. 
Suppose $v=v'j\in V(\zeta)-\{\varepsilon\}$. 
Since the rank of $\zeta[v']$ is greater than zero and thus $\zeta[v']\neq \bot$, 
$v'$ is important in $\zeta$ by the induction hypothesis.  
If $\zeta[v']\in \Delta\cup Y$, $v$ is also important in $\zeta$ 
trivially from the importance of $v'$. Suppose $\zeta[v']=\<q', x_i>$. 
Let $m'$ be the rank of $q'$ in $M$. 
It follows from Statement (i) in the proof of Lemma~\ref{lem:equivalence_MTTFV2MTTC} that 
$\rho(q', l)= j$ for some $l\in [m']$ if and only if $y_j$ appears in $M'_{q'}(s)$ for all $s\in T_\Sigma$. Thus, $v$ is important in $\zeta$ if and only if $\rho(q', l)= j$ 
for some $l\in [m']$. Recall that $\rhs_{M'}(q, \sigma)=\expand_q(\rhs_M(q, \sigma))$ from Definition~\ref{def:MTTFV2MTTC}. By the definition of $\expand_q$, if $\rho(q', l)\neq j$ 
for all $l\in [m']$, we get $\zeta[v]=\bot$. 

Assume that $y_j$ occurs in $\rhs_{M'}(q,\sigma)$ for some $\sigma\in \Sigma$. 
Since $\rhs_{M'}(q, \sigma)=\expand_q(\rhs_M(q, \sigma))$, by the definition of $\expand_q$, 
$y_l$ occurs in $\rhs_M(q,\sigma)$ and $\rho(q, l)=j$ for some $l\in [m]$ where $m$ is the rank of $q$ in $M$.  
Let $q'\in Q$, $\sigma'\in \Sigma, v\in V(\zeta)$ where $\zeta=\rhs_{M'}(q',\sigma)$. 
Suppose that $\zeta[v]=\<q, x_i>$ for some $x_i\in X$. 
Since $\rho(q, l)=j$, it follows from Property~(i) that $vj$ is important in $\zeta$. 
Hence, the $j$th parameter of $q$ is permanent in $M'$. 
\QED
\end{proof}

\begin{lemma}
\label{lem:FV2C_consistent}
Suppose that the nondeleting MTT $M$ has the FV property with $\rho$. Then $\mathcal{E}(M)$ is consistent.
\end{lemma}
\begin{proof}
Let $M'=\mathcal{E}(M)=(Q_e,\Sigma,\Delta,q_0,R_e)$. 
Let $\sigma\in \Sigma^{(k)}$ with $k\geq 0$. Let $q_1,q_2\in Q_e$.  
Let $p_1^{(m_1)}, p_2^{(m_2)}\in Q_e$ and $l\in [k]$. 
For $i=1,2$, let $\zeta_i=\rhs_{M'}(q_i, \sigma)$, $w_i\in V(\zeta_i)$ such that 
$\zeta_i[w_i]=\<p_i, x_l>$. 
Let $j\in [\min\{m_1,m_2\}]$ such that $w_1j$ and $w_2j$ are important in $\zeta_1$ and $\zeta_2$, 
respectively. 
Note that $k>0$ and $m_1=m_2=\tilde{m}$ because of the existence of $w_i$ and $j$.  
From Lemma~\ref{lem:FV2C_PROP}(i), we get 
$\rho(p_1, l_1)=\rho(p_2, l_2)=j$ for some $l_1, l_2\in [\tilde{m}]$. 
Since $\zeta_i = \expand_{q_i}(\zeta_{q_i})$ where $\zeta_{q_i}=\rhs_M(q_i, \sigma)$, from the definition of $\expand_{q}$, 
there exists node $w_i'\in V(\zeta_{q_i})$ such that  $\zeta_{q_i}[w_i']=\<p_i, x_l>$ and $\expand_{q_i}(\zeta_{q_i}/w_i'l_i)=\zeta_i/w_ij$. 
By the FV property of $M$, $(\zeta_{q_1}/w_1'l_1)\Psi^\rho_{q_1}=(\zeta_{q_2}/w_2'l_2)\Psi^\rho_{q_2}$.
Thus, we get $\expand_{q_1}(\zeta_{q_1}/w_1'l_1)=\expand_{q_2}(\zeta_{q_2}/w_2'l_2)$, and thus $\zeta_1/w_1j=\zeta_2/w_2j$. This implies $\Top(\zeta_1/w_1j)=\Top(\zeta_2/w_2j)$. 
\QED
\end{proof}

We shall show that $\mathcal{E}(M)$ satisfies FV-orig if the nondeleting $M$ has the FV property. That is, the resulting ATT obtained 
from $\mathcal{E}(M)$ by the construction $\Omega$ given in [13] is circular. 
Note that this does not hold for consistent MTTs in general. 
For the proof, we give a simplified version of the same construction here
because the style of the definition of MTTs is different with that of [13] and 
we focus on the consistent MTTs with the properties in Lemma~\ref{lem:FV2C_PROP} as inputs. The difference with the original $\Omega$ is just that 
important/non-important nodes in the right-hand side of rules are explicit 
in $\mathcal{E}(M)$, and thus we define $\Omega$ without the recursive functions \textit{SUB} and \textit{DECOMPOSE} in the original $\Omega$. It is straightforward that the following definition of $\Omega$ is equivalent with the original one for $\mathcal{E}(M)$. 

\begin{definition}
\rm
\label{def:MTTC2ATT}
Suppose that the nondeleting $\MTT$ $M$ satisfies the FV property with $\rho$. 
Let $M'=\mathcal{E}(M)=(Q, \Sigma, \Delta, q_0, R)$. 
Let $\tilde{m}$ be the maximum rank of states in $Q$. 
The \emph{attributed tree transducer associated with $M'$}, denoted by $\Omega(M')$, is 
the $\ATT$ $(S, I, \Sigma, \Delta, q_0, R')$ defined as follows:
$S=Q$,
$I=[y_i \mid i\in [\tilde{m}]]$, and
$R'=\bigcup_{\sigma\in \Sigma} R_{\sigma}$.  
For every $\sigma \in \Sigma^{(k)}$ with $k\geq 0$, 
$R_{\sigma}$ is constructed as follows:
\begin{enumerate}
 \item For $q\in Q$, let $\zeta=\rhs_{M'}(q,\sigma)$, 
 \begin{itemize}
  \item  let the rule $q(\pi) \to \Top(\zeta)$ be in $R_{\sigma}$, and
  \item  for every $v\in V(\zeta)$ such that $\zeta[v]=\langle q', x_i\rangle$ where $q'\in Q^{(\tilde{m})}$ and $i\in [k]$, 
            let the rule $y_j(\pi i) \to \Top(\zeta/vj)$ be in $R_{\sigma}$ for each $j\in [\tilde{m}]$ such that $vj$ is important in $\zeta_q$. 
 \end{itemize}

 \item If there is no rule with $\beta(\pi i)$ in the left-hand side for $\beta\in I$ and $i\in [k]$, 
         let the dummy rule $\beta(\pi i) \to \bot$ be in $R_{\sigma}$ where $\bot$ is any label in $\Delta^{(0)}$. 
\end{enumerate}
\end{definition}

Now we show that $\Omega(\mathcal{E}(M))$ is non-circular.  
For this, we give the following technical lemma (Lemma~\ref{lem:same-state-self-nested}) which deduces from the existence of certain paths in the dependency graph of $\Omega(\mathcal{E}(M))$
certain properties about occurrences of state calls in the right-hand sides of $\mathcal{E}(M)$, which the proof of the non-circularity of $\Omega(\mathcal{E}(M))$ (Lemma~\ref{lem:non-circularity-FV}) is based on. 

\begin{lemma}\label{lem:same-state-self-nested}
\rm
Suppose that $M$ is the consistent MTT obtained by $\mathcal{E}$ 
from a nondeleting MTT with the FV property.  
Let $\Omega(M)=(S, I, \Sigma, \Delta, q_0, R')$. 
Let $s$ be an arbitrary tree in $T_\Sigma$. 
\begin{romanenumerate}
\item
For every $q\in S$, and $y_\ell\in I$, 
if there exists a path from $(y_\ell,\varepsilon)$ to $(q,\varepsilon)$ in $D_{\Omega(M)}(s)$,
then the $\ell$th parameter of $q$ is permanent. 
\item
For every $q,q'\in S$ and $i, i'\in [\rank_\Sigma(s[\varepsilon])]$, 
if there exists a nonempty path from $(q,i)$ to $(q',i')$ in $D_{\Omega(M)}(s)$ 
such that 
$\<q', x_{i'}>=\zeta[v']$ for some $\tilde{q}\in Q$ and $v'\in V(\zeta)$ where $\zeta=\rhs_M(\tilde{q}, s[\varepsilon])$,  
then there exists a proper descendant $v$ of $v'$ in $\zeta$ such that $\zeta[v]=\<q, x_{i}>$. 
\end{romanenumerate}
\end{lemma}

\begin{proof}
We prove that (i) and (ii) hold for every $s\in T_\Sigma$ by induction on the structure of $s$ 
in the following way. 
\begin{itemize}
\item Base case. Let $s\in \Sigma^{(0)}$. Since (ii) follows from $\rank_\Sigma(s[\varepsilon])=0$,  
we prove only that (i) holds for $s$ below. 

\item Induction step. Let $s=\sigma(s_1,\ldots,s_k)$ be an arbitrary tree in $T_\Sigma$ such that $k>0$, $\sigma\in \Sigma^{(k)}$, and $s_1,\ldots,s_k\in T_\Sigma$. Under the induction hypothesis of (i), denoted by (IHi) henceforth, for the subtrees of $s$, we first prove that (ii) holds for $s$, 
and then by using (ii) for $s$ as well, we prove that (i) holds for $s$. 
\end{itemize}

\paragraph*{Base case of (i)}
Let $s=\sigma\in \Sigma^{(0)}$. 
Let $q\in S$ and $y_\ell\in I$. 
Assume that there exists a path from $(y_\ell,\varepsilon)$ to $(q,\varepsilon)$ in $D_{\Omega(M)}(s)$. 
Since $\sigma\in \Sigma^{(0)}$, the path consists of only a direct edge from $(y_\ell,\varepsilon)$ to $(q,\varepsilon)$. 
The edge originates from a rule in $R'_{\sigma}$ such that $q(\pi)$ and $y_\ell(\pi)$ occur in the lhs and rhs, respectively. 
From the construction $\Omega(M)$ in Definition~\ref{def:MTTC2ATT}, 
the rhs is $\Top(\rhs_M(q, \sigma))$. 
Since $y_\ell(\pi)$ occurs in the rhs, by the definition of $\Top$, 
$y_\ell$ occurs in $\rhs_M(q, \sigma)$. 
From Lemma~\ref{lem:FV2C_PROP}(ii), the $\ell$th parameter of $q$ is permanent. 

\paragraph*{Induction step of (ii)} 
Let $s=\sigma(s_1,\ldots,s_k)\in T_\Sigma$ with $k>0$ and $\sigma\in \Sigma^{(k)}$. 
Let $q, q'\in S$ and $i, i'\in [k]$. 
Hereafter, we abbreviate $\rhs_M(p, \sigma)$ as $\zeta_p$ for $p\in Q$. 
Assume that there exists a nonempty path $w$ from $(q,i)$ to $(q',i')$ in $D_{\Omega(M)}(s)$, and that there exists $\tilde{q}\in Q$ and $v'\in V(\zeta_{\tilde{q}})$ such that $\<q', x_{i'}>=\zeta_{\tilde{q}}[v']$. 

We can regard the path $w$ as a sequence of nodes of $D_{\Omega(M)}(\sigma)$ such that 
every two consecutive nodes are connected by an edge of  $D_{\Omega(M)}(\sigma)$, or 
a path of a subgraph $D_{\Omega(M)}(s_i)$ for some $i\in [k]$. 
Let $\leftarrow_\varepsilon$ and $\leftarrow_i$ denote connections by an edge of $D_{\Omega(M)}(\sigma)$ and 
a path of a subgraph $D_{\Omega(M)}(s)/i$ for $i\in [k]$, respectively. 
Then, let $q_1=q'$, $i_1=i'$, $q_n=q$, and $i_n=i$, and we can write the path $w$ as follows: 
\begin{multline*}
(q', i')=(q_1,i_1)\leftarrow_{i_1} (y_{\ell_1},i_1)  \leftarrow_\varepsilon (q_2,i_2) \\\leftarrow_{i_2} \cdots \leftarrow_\varepsilon (q_{n-1},i_{n-1})\leftarrow_{i_{n-1}} (y_{\ell_{n-1}},i_{n-1})  \leftarrow_\varepsilon (q_n,i_n)=(q,i)
\end{multline*}
for some $n>1$ where $i_j\in [k]$ and $q_j \in S$ for all $j\in [n]$, and $y_{\ell_j} \in I$ for all $j\in [n-1]$. For such $w$, now we prove that there exists a proper descendant $v$ of $v'$ in $\zeta_{\tilde{q}}$ such that $\zeta_{\tilde{q}}[v]=\<q, x_i>$.
To achieve this, 
we show that 
for every $c$ where $2\leq c \leq n$,  
there exists a proper descendant $v$ of $v'$ in $\zeta_{\tilde{q}}$ such that $\zeta_{\tilde{q}}[v]=\<q_c, x_{i_c}>$. 
It can be shown by induction on $c$. 
Let us denote the induction hypothesis of the inner induction on $c$ by (IHc). 

(Base case) 
Let $c=2$. We have $(q_1,i_1)\leftarrow_{i_1} (y_{\ell_1},i_1) \leftarrow_\varepsilon (q_2,i_2)$. 
Since the edge $(q_{1},i_{1})\leftarrow_{i_{1}} (y_{\ell_{1}},i_{1})$ is in $D_{\Omega(M)}(s)/i_{1}$, 
an edge from $(y_{\ell_{1}},\varepsilon)$ to $(q_{1},\varepsilon)$ is in $D_{\Omega(M)}(s/i_{1})$. 
By (IHi) to the edge, the $\ell_1$th parameter of $q_1$ is permanent. 
Recall that $\zeta_{\tilde{q}}[v']=\<q',x_{i'}>(=\<q_1,x_{i_1}>)$. 
Thus, $v'\ell_1$ is important in $\zeta_{\tilde{q}}$. 
Moreover, the edge from $(q_2,i_2)$ to $(y_{\ell_{1}},i_{1})$ originates from a rule in $R'_{\sigma}$ 
such that $y_{\ell_{1}}(\pi i_{1})$ and $q_{2}(\pi i_2)$ occur in the lhs and rhs, respectively. 
From the construction of Definition~\ref{def:MTTC2ATT}, 
for some $\hat{q}, \check{q}\in Q$, $u\in V(\zeta_{\hat{q}})$, 
the rhs is $\Top(\zeta_{\hat{q}}/u\ell_{1})$,  
$\zeta_{\hat{q}}[u]=\<\check{q}, x_{i_{1}}>$, and 
$u\ell_1$ is important in $\zeta_{\hat{q}}$. 
Since $\zeta_{\tilde{q}}[v']=\<q_1, x_{i_1}>$,  $\zeta_{\hat{q}}[u]=\<\check{q}, x_{i_{1}}>$, 
and $v'\ell_1$ and $u\ell_1$ are important in $\zeta_{q}$ and $\zeta_{\hat{q}}$, respectively, 
we get $\Top(\zeta_{\tilde{q}}/v'\ell_{1})=\Top(\zeta_{\hat{q}}/u\ell_{1})$ by the consistency of $M$. 
 Since $\<q_2, x_{i_2}>$ occurs in $\Top(\zeta_{\hat{q}}/u\ell_1)$, it also occurs in $\Top(\zeta_{\tilde{q}}/v'\ell_1)$.
By the definition of $\Top$, there exists a proper descendant $v$ of $v'$ in $\zeta_{\tilde{q}}$ such that $\zeta_{\tilde{q}}[v]=\<q_2, x_{i_2}>$. 

(Induction step) 
By applying (IHc) to the sub-path $(q_1,i_1)\leftarrow_{i_1} \cdots \leftarrow_{\varepsilon} (q_{c-1},i_{c-1})$,  
it holds that there exists a descendant $v_{c-1}$ of $v'$ in $\zeta_{\tilde{q}}$ such that $\zeta_{\tilde{q}}[v_{c-1}]=\<q_{c-1}, x_{i_{c-1}}>$. 
From the part $(q_{c-1},i_{c-1})\leftarrow_{i_{c-1}} (y_{\ell_{c-1}},i_{c-1}) \leftarrow_\varepsilon (q_c,i_c)$, with a similar discussion as (Base case),  
we can show that there exists a proper descendant $v_c$ of $v_{c-1}$ in $\zeta_{\tilde{q}}$ such that $\zeta_{\tilde{q}}[v_c]=\<q_c, x_{i_c}>$. 
Therefore, the node $v_c$ is a proper descendant of $v'$ such that $\zeta_{\tilde{q}}[v_c]=\<q_c, x_{i_c}>$.

\paragraph*{Induction step of (i)}
Let $s=\sigma(s_1,\ldots,s_k)\in T_\Sigma$ with $k>0$ and $\sigma\in \Sigma^{(k)}$. 
Let $q\in S$ and $\ell\in I$. 
Assume that there exists a path $w$ from $(y_\ell,\varepsilon)$ to $(q,\varepsilon)$ in $D_{\Omega(M)}(s)$. 
Hereafter, we abbreviate $\rhs_M(p, \sigma)$ as $\zeta_{p}$ for $p\in Q$. 

For the case that $w$ consists of only a direct edge from $(y_\ell,\varepsilon)$ to $(q,\varepsilon)$, 
we can complete this case with the same discussion as the base case of the proof of Lemma~\ref{lem:same-state-self-nested}.

For the other case, we can assume that the path $w$ goes from $(y_\ell,\varepsilon)$ to $(q,\varepsilon)$ via at least one synthesized attribute node. 
As in the proof of Induction step of (ii), we can regard the path $w$ as 
$(q,\varepsilon)\leftarrow_\varepsilon (q_1,i_1)\leftarrow^* (q_n,i_n)\leftarrow_{i_n} (y_{\ell_n},i_n) \leftarrow_\varepsilon (y_\ell,\varepsilon)$
where $i_1, i_n\in [k]$, $q, q_1, q_n \in S$, $y_\ell, y_{\ell_n}\in I$, and $\leftarrow^*$ is 
the reflexive transitive closure of $(\bigcup_{i\in [k]} \leftarrow_i)\cdot \leftarrow_\varepsilon$. 
First, we show from the sub-path $(q,\varepsilon)\leftarrow_\varepsilon (q_1,i_1)\leftarrow^* (q_n,i_n)$ that 
there exists a node $v_n$ in $\zeta_{q}$ such that $\zeta_{q}[v_n]=\<q_{n}, x_{i_{n}}>$.
The edge $(q,\varepsilon)\leftarrow_\varepsilon (q_1,i_1)$ originates from a rule such that $q(\pi)$ and $q_1(\pi i_1)$ occur in the lhs and rhs. 
From the construction of Definition~\ref{def:MTTC2ATT}, the rhs is $\Top(\zeta_q)$. 
By the definition of $\Top$ and the fact that $q_1(\pi i_1)$ occurs in $\Top(\zeta_q)$,  
$\zeta_q[v_1]=\<q_1,x_{i_1}>$ for some $v_1\in V(\zeta_q)$. 
From the sub-path $(q_1,i_1)\leftarrow^* (q_n,i_n)$ in $w$, 
we can obtain the fact that there exists a descendant $v_n$ of $v_1$ in $\zeta_{q}$ such that $\zeta_{q}[v_n]=\<q_{n}, x_{i_{n}}>$ as follows. 
If the sub-path $(q_1,i_1)\leftarrow^* (q_n,i_n)$ in $w$ is empty, we have $q_n=q_1$ and $i_n=i_1$,   
and thus $\zeta_{q}[v_1]=\<q_{n}, x_{i_{n}}>$. 
Otherwise, since (ii) holds for $s$ under (IHi) for $s_1,\ldots,s_k$ (by the induction step of (ii)), by the existence of the nonempty path from $(q_n,i_n)$ to $(q_1,i_1)$ with 
$\zeta_q[v_1]=\<q_1,x_{i_1}>$, 
there exists a proper descendant $v_n$ of $v_1$ in $\zeta_{q}$ such that $\zeta_{q}[v_n]=\<q_{n}, x_{i_{n}}>$. 
Next, by (IHi) to the sub-path $(q_n,i_n) \leftarrow_{i_n} (y_{\ell_n},i_n)$, the $\ell_n$th parameter of $q_n$ is permanent. Thus, $v_n\ell_n$ is important in $\zeta_q$. 
The edge $(\ell_n,i_n)\leftarrow_\varepsilon (\ell,\varepsilon)$ originates from a rule such that $\ell_n(\pi i_n)$ and 
$\ell(\pi)$ occur in the lhs and rhs, respectively. 
From the construction of Definition~\ref{def:MTTC2ATT}, 
for some $q', q''\in Q$ and $u\in V(\zeta_{q'})$, 
the rhs is $\Top(\zeta_{q'}/u\ell_n)$ and $\zeta_{q'}[u]=\<q'', x_{i_n}>$ and $u\ell_n$ is important in $\zeta_{q'}$. 
Since $\zeta_q[v_n]=\<q_n,x_{i_n}>$, $\zeta_{q'}[u]=\<q'',x_{i_n}>$, and $v_n\ell_n$ and $u\ell_n$ are important in $\zeta_q$ and $\zeta_{q'}$, we get $\Top(\zeta_q/v_n\ell_n)=\Top(\zeta_{q'}/u\ell_n)$ by the consistency of $M$. 
since $y_\ell$ appears in $\Top(\zeta_{q'}/u\ell_n)$, 
$y_\ell$ appears in $\Top(\zeta_q/v_n\ell_n)$. 
Thus, $y_\ell$ appear in $\zeta_q/v_n\ell_n$. 
From Lemma~\ref{lem:FV2C_PROP}(ii),  
the $\ell$th parameter of $q$ is permanent. 
\QED
\end{proof}

\begin{lemma}\label{lem:non-circularity-FV}
\rm
If the nondeleting MTT $M$ has the FV property then $\Omega(\mathcal{E}(M))$ is non-circular.
\end{lemma}
\begin{proof}
Let $M'=\mathcal{E}(M)$. 
Let $A=\Omega(M')=(S, I, \Sigma, \Delta, q_0, R')$. 
The proof is by contradiction. Assume that $A$ is circular. 
Then there exists a tree $s\in T_\Sigma$ such that $D_{A}(s)$ has a cycle. 
Let $u\in V(s)$ be a node such that $s/u$ is a minimal subtree that includes the cycle. 
Then $s/u\in T_\Sigma-\Sigma^{(0)}$ because no cycle can be made in $D_{A}(\sigma)$ for any 
$\sigma\in \Sigma^{(0)}$. 
Let $\sigma=s[u]\in \Sigma^{(k)}$ with $k>0$. 
There is a cycle from some synthesized attribute node $(q,i)$ to itself via at least one inherited attribute node 
in $D_{A}(s/u)$ for some $q\in S$ and $i\in [k]$. 
Since there exists an outgoing edge of $(q,i)$, by the construction, 
there exist $\tilde{q}\in Q$ and $v'\in V(\zeta)$ such that $\zeta[v']=\<q, x_i>$ where 
$\zeta=\rhs_{M'}(\tilde{q},\sigma)$. 
From Lemma~\ref{lem:same-state-self-nested}, there exists a proper descendant $v$ of $v'$ in $\zeta$ such that $\zeta[v]=\<q, x_i>$. 
Then, let us consider the loop end of state calls in $\zeta$ involved by the existence of such nodes $v$ and $v'$, that is, we can choose distinct two nodes $u$ and $u'$ in $\zeta$ such that $\zeta[u]=\zeta[u']=\<q',x_j>\in \<Q,X>$, $u'$ is a descendant of $uj$ for some $j\in [\tilde{m}]$, any call $\<q'', x_l>\in \<Q,X>$ over the path from $u$ to $u'$ does not occur in $\zeta/u'j$. This implies $\Top(\zeta/uj)\neq \Top(\zeta/u'j)$. 
On the other hand, $\zeta[uj]\neq \bot$. From Lemma~\ref{lem:FV2C_PROP}(i), $uj$ is important in $\zeta$ and thus $y_j$ occurs in $\rhs_{M'}(q',\sigma')$ for some $\sigma'\in \Sigma$. Since the $j$th parameter of $q'$ is permanent from Lemma~\ref{lem:FV2C_PROP}(ii), 
$u'j$ is also important in $\zeta$. Since $M'$ is consistent from Lemma~\ref{lem:FV2C_consistent}, $\Top(\zeta/uj)=\Top(\zeta/u'j)$, which is a contradiction. 
\QED
\end{proof}


For the nondeleting MTT $M$ with the FV property, 
from Lemmas~\ref{lem:equivalence_MTTFV2MTTC}, \ref{lem:FV2C_consistent}, and \ref{lem:non-circularity-FV}, 
$\mathcal{E}(M)$ satisfies the FV-orig property (i.e., \emph{attributed-like} in \cite{FV99}). 
Lemma~3.18 in \cite{FV99} says that for an MTT with the FV-orig property, 
there is an equivalent ATT. 
From that, we obtain the following lemma and corollary. 
\begin{lemma}
\rm
$\CMTTFV\subseteq \CATT$. 
\end{lemma}
\begin{corollary}\label{cor:MTTRF-ATTR}
\rm
$\CMTTRFV\subseteq \CATTR$ and $\CMTTUFV\subseteq \CATTU$. 
\end{corollary}

Next, we show that the converse of the inclusion of Corollary~\ref{cor:MTTRF-ATTR}.
For this, we will show $\CMTTC\subseteq \CMTTRFV$. 
Before that, let us consider an example. 
For the consistent MTT $M_c$ of Figure~\ref{fig:circ}, 
let us construct its nondeleting normal form $B\com M'_c$. 
We will give the details of the construction in Lemma~\ref{lem:FVo-ourFV1}, 
and the normal form is in $\CMTTRFV$. 
We first define the bottom-up relabeling $B$ that realizes the look-ahead:
\[
\begin{array}{lr}
e\to p_1(e)&\text{for }p_1=\{ q\mapsto \{1\}, q_1\mapsto\emptyset, q_2\mapsto \{1\}\}\\
e'\to p_2(e')&\text{for }p_2=\{ q\mapsto \emptyset, q_1\mapsto\{ 1\}, q_2\mapsto \{1\}\}\\
\multicolumn{2}{l}{\sigma(p(x_1),p'(x_2))\to p_3([\sigma,p,p'](x_1,x_2)) \text{~~~for }
p_3=\{ q\mapsto\emptyset, q_1\mapsto\emptyset, q_2\mapsto\emptyset\}}
\end{array}
\]
where $p,p'\in\{p_1,p_2,p_3\}$.
The rules of the MTT $M_c'$ are as follows:
\[
\begin{array}{llcl}
r_1:&\la (q_0,\emptyset), [\sigma,p_1,p_1/p_2](x_1,x_2)\ra&&\\
&\multicolumn{3}{r}{\to~\la (q,\{1\}),x_1\ra(\langle  (q_2,\{1\}),x_2\ra(\la (q_1,\emptyset),x_1\ra))}\\
r_2:&\la (q_0,\emptyset), [\sigma,p_1,p_3](x_1,x_2)\ra&\to& \la (q,\{1\}),x_1\ra(\langle (q_2,\emptyset),x_2\ra)\\
r_3:&\la (q_0,\emptyset), [\sigma,p_2/p_3,\_](x_1,x_2)\ra&\to& \la (q,\emptyset),x_1\ra\\
&\la (\_,\emptyset),\_\ra&\to& \#\\
&\la (\_,\{1\}),\_\ra(y_1)&\to& y_1.
\end{array}
\]
Here the symbol ``$\_$'' denotes any state, or any input symbol (and look-ahead combination)
and the notation $p_1/p_2$ means that for both input symbols rules 
with the same right-hand side exist.
As the reader may verify, this transducer indeed has the FV property (with
$\rho((\_,\{1\}))=1$).
Note how the bottom-most symbol of $\rhs_{M_c}(q_0,\sigma)$, i.e., the $\$$-labeled leaf
does \emph{not} occur in any rule of $M_c'$.
This is because only three possibilities exist how the first rule $r$ of $M_c$ is evaluated
on an input tree $s$ (as expressed by the rules $r_1,r_2,r_3$ of $M_c'$): 
\begin{enumerate}
\item if $s=[\sigma,p_1,p_1/p_2](e, e/e')$, then $r$ evaluates to $\<q_1,e>$ (viz. $r_1$)
\item if $s=[\sigma,p_1,p_3](e, \sigma(\dots))$, then $r$ evaluates to $\<q_2,\sigma(\dots)>$ (viz. $r_2$)
\item if $s=[\sigma,p_2/p_3,\_](e'/\sigma(\dots), \dots)$, then $r$ evaluates to $\<q,e'/\sigma(\dots)>$ (viz. $r_3$).
\end{enumerate}
In each case, the MTT $M_c$ finally applies a rule with $\#$ in the right-hand side, i.e.,
a \emph{deleting} rule.

\begin{lemma}\label{lem:FVo-ourFV1}
\rm
$\CATT\subseteq \CMTTC\subseteq \CMTTRFV$. 
\end{lemma}
\begin{proof}
It is well known that for any $\ATT$ there exists an equivalent consistent $\MTT$~\cite{FV99}.
It should be noted that in Theorem~5.11 of~\cite{EM99} it is proved that
$\CATT\subseteq\CMTTR$. The MTT constructed in the proof of that theorem in fact
has the FV property, so we could simply use that construction. 
However, to make this paper more self-contained we prefer to give another proof, purely in terms of MTTs. 

Let $M=(Q, \Sigma, \Delta, q_0, R)$ be a consistent $\MTT$. 
We show that the nondeleting normal form $M'$ of $M$ obtained by the construction 
in the proof of Lemma~6.6 of~\cite{EM99} has the FV property. 
We give the construction again, in our setting.
  
We denote by $P$ the set of all functions which associate with every $q\in Q^{(m)}$ a subset of $[m]$. 
For a subset $I$ of $\mathbb{N}$ we denote by 
$|I|$ the cardinality of $I$ and by 
$I(j)$ the $j$th element of $I$ with respect to $<$. 
We define the $\BREL$
$B=(P,\Sigma,\Sigma',P,R_B)$, 
where $\Sigma'=\{[\sigma, p_1,\ldots,p_k]\mid \sigma\in \Sigma^{(k)}, p_1,\ldots,p_k\in P\}$. 
For $\sigma\in \Sigma^{(k)}$ and 
$p_1,\ldots,p_k\in P$, let the rule 
\[
\sigma(p_1(x_1),\ldots,p_k(x_k))\to p_0([\sigma,p_1,\ldots,p_k](x_1,\ldots,x_k))
\]
be in $R_B$, where 
for every $q\in Q^{(m)}$, $p_0(q)=\oc(\rhs_M(q,\sigma))$ and for $t\in T_{\<Q,X_k>\cup \Delta}(Y_m)$, 
\[
 \oc(t) = 
\begin{cases}
 \{j\} & \text{if $t=y_j\in Y_m$} \\
 \bigcup_{i=1}^{l} \oc(t_i) & \text{if $t=\delta(t_1,\ldots,t_l)$ where $\delta\in \Delta^{(l)}$} \\
 \bigcup_{j\in p_i(r)} \oc(t_j) & \text{if $t=\<r,x_i>(t_1,\ldots,t_l)$ where $\<r,x_i>\in \<Q, X_k>^{(l)}$.} 
\end{cases}
\]
We define the transition function $h$ such that $h_\sigma(p_1,\ldots,p_k)=p_0$ if rule  $\sigma(p_1(x_1),\ldots,p_k(x_k))\to p_0([\sigma,p_1,\ldots,p_k](x_1,\ldots,x_k))$ is in $R_B$. 

Let $M'=(Q', \Sigma', \Delta \cup \{d^{(2)}\}, (q_0, \emptyset), R')$ be an $\MTT$  
where $Q'=\{(q, I)^{(|I|)}\mid q\in Q^{(m)}, I\subseteq [m]\}$. 
For every $(q,I)\in Q'$, $\sigma \in \Sigma^{(k)}$, $k\geq 0$, and $p_1,\ldots, p_k\in P$, let the rule
\[
\<(q,I),[\sigma, p_1,\ldots, p_k](x_1,\dots,x_k)>(y_1,\ldots,y_{|I|})\to \zeta
\]
be in $R'$, where $p_0=h_\sigma(p_1,\ldots,p_k)$
and for $I\neq p_0(q)$ let $\zeta=y_1$ if $|I|=1$ and 
otherwise let 
$\zeta=d(y_1,d(y_2,\ldots, d(y_{|I|-1}, y_{|I|})\cdots))$;
if $I=p_0(q)$
let $\zeta=\rhs_M(q,\sigma)\Theta_{\vec{p}}\theta_I$, 
where $\Theta_{\vec{p}}$ with $\vec{p}=\<p_1,\ldots,p_k>$ denotes the substitution
\[
[\![\<r,x_i>\leftarrow \<(r,I_r),x_i>(y_{I_r(1)}, \ldots, y_{I_r(n)})
\mid \<r,x_i>\in \<Q,X_k>, I_r=p_i(r), n=|I_r|]\!]
\]
and $\theta_I=[y_{I(j)}\leftarrow y_j\mid j\in [|I|]]$. 
We let $\rho((q,I), j) = I(j)$ for every $(q, I)\in Q'$ and $j\in [|I|]$.

\begin{claim}
The following two properties hold. 
\begin{enumerate} 
\item \label{claim:origin-delete}
For every $(q, I)\in Q'$, $\sigma\in \Sigma^{(k)}$, $k\geq 0$, and $p_1,\ldots, p_k\in P$, 
if  $\rhs_{M'}((q,I), [\sigma, \vec{p}])$, where $\vec{p}$ is the sequence $p_1,\ldots,p_k$, is not the dummy right-hand side, then 
there is an injective mapping $\psi: V(\zeta) \to V(\xi)$ 
where $\zeta=\rhs_{M'}((q,I), [\sigma, \vec{p}])$ and $\xi=\rhs_M(q, \sigma)$ 
such that for every $w\in V(\zeta)$, 
\begin{enumerate}
 \item $\zeta/w=(\xi/\psi(w))\Theta_{\vec{p}}\theta_I$ and $\psi(w)$ is important in $\xi$ 
 for some trees $s_1,\ldots, s_k\in T_\Sigma$ such that $B(\sigma(s_1,\ldots,s_k))[\varepsilon]=[\sigma,\vec{p}]$, and

 \item if $\zeta[w]=\<(q', I'), x_i>\in \<Q', X>$ then $\psi(wj)=\psi(w)\rho((q',I'),j)$ for every $j\in [|I'|]$. 
\end{enumerate}

\item \label{claim:top-extension}
Let $\xi_1=\rhs_M(q_1, \sigma)$ and $\xi_2=\rhs_M(q_2, \sigma)$ for states $q_1,q_2\in Q$. 
Let $p_1,\ldots,p_k\in P$ and $\vec{p}$ be the sequence $p_1,\ldots,p_k$. 
Assume that neither $\rhs_{M'}(q_1, [\sigma, \vec{p}])$ nor $\rhs_{M'}(q_2, [\sigma,\vec{p}])$ is the dummy right-hand side. 
Let $v_1\in V(\xi_1)$ and $v_2\in V(\xi_2)$. 
Assume that $v_1$ and $v_2$ are important in $\xi_1$ and $\xi_2$, respectively, for some $s_1, \ldots, s_k\in T_\Sigma$ 
such that $B(\sigma(s_1,\ldots,s_k))[\varepsilon]=[\sigma,\vec{p}]$. Then, $\Top(\xi_1/v_1)=\Top(\xi_2/v_2)$ implies $(\xi_1/v_1)\Theta_{\vec{p}}=(\xi_2/v_2)\Theta_{\vec{p}}$. 
\end{enumerate}
\end{claim}
%
Now we show that $M'$ has the FV property with $\rho$ by using the above claim. 
Let $(q_1, I_1), (q_2, I_2)\in Q'$, $\sigma\in \Sigma^{(k)}$, $k\geq 0$, and $p_1,\ldots, p_k\in P$. 
Let $\zeta_i=\rhs_{M'}((q_i,I_i), \sigma')$  for $i\in [2]$ where $\sigma'=[\sigma, p_1,\ldots,p_k])$
Let $w_1\in V(\zeta_1)$ and $w_2\in V(\zeta_2)$ such that $\zeta_1[w_1]=\<(q,I), x_i>$ and $\zeta_2[w_2]=\<(q',I'),x_i>$. 
Let $j\in [\rank_{M'}((q,I))]$ and $j'\in [\rank_{M'}((q',I'))]$ such that $\rho((q,I),j)=\rho((q',I'),j')=j''$. 
Let $\xi_i=\rhs_M(q_i, \sigma)$ for $i\in [2]$. 
Let $\psi_i$ be the mapping defined by the property~\ref{claim:origin-delete} of the claim.  
By property~\ref{claim:origin-delete}, $\psi_1(w_1j)$ and $\psi_2(w_2j')$ are important in $\xi_1$ and $\xi_2$, respectively, for some trees $s_1,\ldots,s_k$ such that $B(\sigma(s_1,\ldots,s_k))[\varepsilon]=[\sigma,\vec{p}]$. 
In addition, we get $\zeta_1/w_1j=(\xi_1/\psi_1(w_1j))\Theta_{\vec{p}}\theta_{I_1}$ and 
$\zeta_2/w_2j'=(\xi_2/\psi_2(w_2j'))\Theta_{\vec{p}}\theta_{I_2}$. 
Since $\Psi^\rho_{(q_i, I_i)}=\theta^{-1}_{I_i}$ for each $i\in [2]$, 
$(\zeta_1/w_1j)\Psi^\rho_{(q_1, I_1)}=(\xi_1/\psi_1(w_1j))\Theta_{\vec{p}}$ and 
$(\zeta_2/w_2j')\Psi^\rho_{(q_2, I_2)}=(\xi_2/\psi_2(w_2j'))\Theta_{\vec{p}}$. 
By the fact that $M$ is consistent, 
$\Top(\xi_1/\psi(w_1j))=\Top(\xi_2/\psi(w_2j'))$. 
By the second property of the claim, 
we get $(\xi_1/\psi_1(w_1j))\Theta_{\vec{p}}=(\xi_2/\psi_2(w_2j'))\Theta_{\vec{p}}$. 
Thus, $(\zeta_1/w_1j)\Psi^\rho_{(q_1, I_1)}=(\zeta_2/w_2j')\Psi^\rho_{(q_2, I_2)}$. 
Hence, $M'$ has the FV property with $\rho$. 
\QED
\end{proof}

Before we state our main theorem of this section, let us compare the circular ATT obtained 
by the construction of \cite{FV99}, and the ATT obtained from the nondeleting $M'_c$ with 
the FV property given before Lemma~\ref{lem:FVo-ourFV1}. 
We apply the construction of Definitions~\ref{def:MTTFV2MTTC} and \ref{def:MTTC2ATT} to the MTT $M_c'$
and generate the ATT $A_c=\Omega(\mathcal{E}(M_c'))$.
We have $S=\{ q_00, q0, q1, q_10, q_20, q_21 \}$; these are the states of $M_c'$ 
(in condensed notation). The inherited attributes of $A_c$ are $I=\{1\}$ ($y_1$ is abbreviated as $1$).
The set $R_{[\sigma,p_1,p_1/p_2]}$ contains the rule
\[
q_00(\pi)\to \Top(\expand_{q_00}(\rhs_{M_c'}(q_00,[\sigma,p_1,p_1/p_2])))=q1(\pi 1).
\]
Since the right-hand side of $r_1$ has two parameter trees, we obtain two more rules:
\[
\underbrace{\rho(q1,1)}_{=1}(\pi 1)\to\Top(\expand_{q_00}(\la q_21,x_2\ra(\la q_10,x_1\ra)))=q_21(\pi 2)
\]
and 
$1(\pi 2)\to\Top(\expand_{q_00}(\la q_10,x_1\ra))=q_10(\pi 1)$.
The dependency graph for the input tree $\sigma(e,e)$ 
depicted in Figure~\ref{fig:Mc1} shows the full set
of rules $R_{[\sigma,p_1,p_1/p_2]}$.
Note that for the input tree $\sigma(e',e)$ which causes the construction
of~\cite{FV99} for $M_c$ to generate a circular ATT, the corresponding dependency 
graph for $A_c$ looks rather innocent, see Figure~\ref{fig:Mc2}.
\begin{figure}[t]
\centering
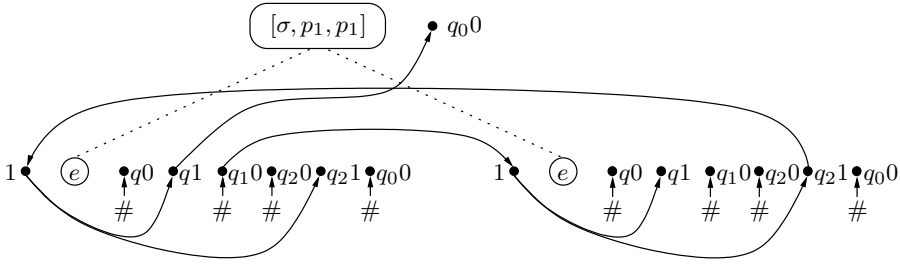
\caption{The dependency graph of the ATT $A_c$ for the input tree $[\sigma,p_1,p_1](e,e)$}
\label{fig:Mc1}
\end{figure} 
\begin{figure}[t]
\centering
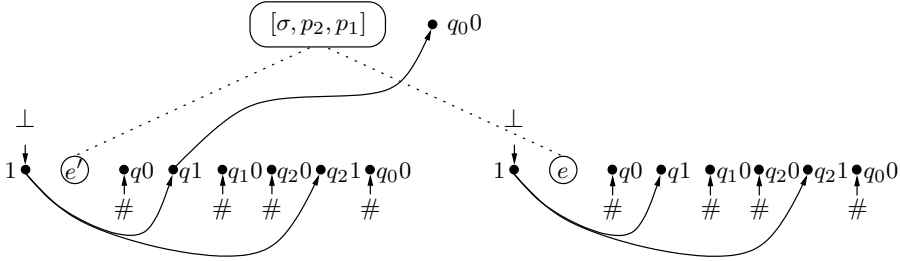
\caption{The dependency graph of the ATT $A_c$ for the input tree $[\sigma,p_2,p_1](e',e)$}
\label{fig:Mc2}
\end{figure}

%


By Corollary~\ref{cor:MTTRF-ATTR} and Lemma~\ref{lem:FVo-ourFV1}, 
and the facts that $\CBREL\subseteq\CTRREL$ and that
 $\CBREL$ and $\CTRREL$ are closed under composition
we obtain the main result of this section.


\begin{theorem}\label{thm:ATTR=MTTRF}
\rm
$\CATTR=\CMTTRC=\CMTTRFV$ and $\CATTU=\CMTTUC=\CMTTUFV$. 
\end{theorem}

%
%


\section{Dynamic FV Property}

We present a new property that characterizes ATT with
regular look-around in terms of MTTs with regular look-ahead. 
The \emph{dynamic FV} property is a strict generalization of the FV property. 
The idea is to require that during any derivation
the argument trees of any state should be semantically equal. 
In the following, let $M=(Q,\Sigma,\Delta,q_0,R)$ be an $\MTT$. 

\begin{definition}
\rm
Let $q\in Q$, $s\in T_\Sigma$, and $u\in V(s)$. 
The set of \emph{call trees of $q$ at $u$ on $s$}
is defined as $\ct_M(s,u,q)=\{\xi/v \mid v\in V(\xi), \xi[v]=\<q,x>\}$
where $\xi=M_{q_0}(s[u\leftarrow x])$. 
\end{definition}
\begin{definition}
\rm
The MTT $M$ has the \emph{dynamic FV property} on a set $L$ of trees if 
for every $s\in L$, $u\in V(s)$, $q\in Q^{(m)}$, $j\in [m]$, 
and $\xi_1, \xi_2\in\ct_M(s, u, q)$, 
$\xi_1/j[\![s/u]\!]=\xi_2/j[\![s/u]\!]$ where $[\![s/u]\!]=[\![\<q', x>\leftarrow M_{q'}(s/u) \mid q'\in Q]\!]$. 
We say that an $\MTTRDFV$ (resp. $\CMTTUDFV$) $E\com M$, where $E$ is a relabeling and $M$ is an $\MTT$, 
has the dynamic FV property if $M$ has the dynamic FV property on the image $E(T_\Sigma)$. 
\end{definition}

We denote by $\CMTTRDFV$ and $\CMTTUDFV$ the class of all translations realized by $\MTTR$s and $\MTTU$s 
with the dynamic FV property, respectively. 
As example we consider the MTT $M_{\text{dyn}}$ which does not have the FV property with any parameter renaming 
mapping $\rho$ but has the dynamic FV property. 
These are the rules of $M_{\text{dyn}}$:
\[
\begin{array}{lcllcl}
\<q_0,a(x_1)>&\to&f(\<q_1,x_1>(\<q_2,x_1>),\<q_1,x_1>(\<q_3,x_1>(e))) \\ 
\<q_1,a(x_1)>(y_1)&\to&a(\<q_1,x_1>(b(y_1))) \\
\<q_2,a(x_1)>&\to&a(a(\<q_2,x_1>))\\
\<q_3,a(x_1)>(y_1)&\to&a(\<q_3,x_1>(a(y_1)))\\
\<q_0,e>&\to&e\\
\<q_2,e>&\to&e\\
\<q_1,e>(y_1)&\to&y_1\\
\<q_3,e>(y_1)&\to&y_1
\end{array}
\]

The MTT $M_{\text{dyn}}$ translates trees of the form $a^n(e)$ to trees of the form 
$f(t, t)$ where $t=a^{n-1}b^{n-1}a^{2n-2}(e)$. 
Clearly, $M$ does not satisfy the FV property. 
On the other hand, the two argument trees of $q_1$ always evaluate to the same trees
and hence $M$ satisfies the dynamic FV property.
We show that the dynamic FV property is a generalization of the FV property. 
\begin{lemma}\label{lem:FVtoDFV}
\rm
$\CMTTFV \subseteq \CMTTDFV$. 
\end{lemma}
\begin{proof}
Let $M=(Q,\Sigma,\Delta,q_0,R)$ be an $\MTT$ that has the FV property with $\rho$.  
We give the following claim to prove the lemma.
\begin{claim}\label{cl:FV->DFV}
For every $s\in T_\Sigma$, $u\in V(s)$, $q^{(m)}\in \Sts_{M}(\{q_0\},s,u)$, $m\geq 0$, $j\in [m]$, 
and $t_1,t_2\in \ct_{M}(s, u, q)$, it holds that $t_1/j=t_2/j$. 
\end{claim}
To show that the claim holds, we prove  
the following property by induction on the length of $u$:  
for every $u\in \mathbb{N}^*$, $s\in T_\Sigma$ such that $u\in V(s)$, $q^{(m_1)}_1, q^{(m_2)}_2\in \Sts_{M}(\{q_0\}, s, u)$
with $m_1,m_2\geq 0$, $t_1\in \ct_{M}(s, u, q_1)$, $t_2\in \ct_{M}(s, u, q_2)$, $j_1\in [m_1]$, and $j_2\in [m_2]$, 
if $\rho(q_1, j_1)=\rho(q_2, j_2)$ then $t_1/j_1=t_2/j_2$. 
Since $\rho(q, j_1)=\rho(q,j_2)$ implies $j_1=j_2$, we obtain the claim.

It is clear from Claim~\ref{cl:FV->DFV} that $M$ has the dynamic FV property. 
\QED
\end{proof}

We say that $M$ is \emph{nonerasing} if for every $q\in Q$ and $\sigma\in \Sigma$, $\rhs_M(q,\sigma)\notin Y$. 
The nondeleting and nonerasing normal form can be obtained by the construction given in Lemmas~6.6 and 7.11 in \cite{EM99}. 
The construction preserves the dynamic FV property. 

\begin{lemma}\label{lem:NonDelPrvDFV}
\rm
Let $L\subseteq T_\Sigma$ and let $M=(Q,\Sigma,\Delta,q_0,R)$ be an $\MTT$ that has the dynamic FV property on $L$. 
Let $B\com M'$ be the nondeleting $\MTTR$ obtained from $M$ by the construction given in the proof of Lemma~\ref{lem:FVo-ourFV1}. 
Then $M'$ has the dynamic FV property on $B(L)$. 
\end{lemma}
\begin{proof}
Let $s\in L$ and $u\in V(s)$. Let $s'=s[u\leftarrow x]$ and $s''=B(s)[u\leftarrow x]$. 
Let $p_0=h_\sigma(p_1,\ldots,p_k)$ and $(\sigma, p_1,\ldots,p_k)=B(s)[u]$. 
Since the construction of the nondeleting normal form is the same as one given in Lemma 6.6 in \cite{EM99},   
we get the following claim from Claims 2 and 3 in the proof of Lemma 6.6 in \cite{EM99}.  
\begin{claim} The following two properties holds. 
\begin{enumerate}
\item\label{cl:del-prop1}
$M'(s'') = M(s')[\![\<q, x>\leftarrow \<(q, I),x>(y_{I(1)},\ldots, y_{I(n)})\mid I=p_0(q), n=|I|]\!]$. 
\item\label{cl:del-prop2}
$M'_{(q,I)}(B(s)/u)\theta_I=M_q(s/u)$ for every $q\in Q$ where $I=p_0(q)$. 
\end{enumerate}
\end{claim}
From the above claim, we get the injective mapping $\psi: V(M'(s''))\to V(M(s'))$ such that 
\begin{itemize}
 \item $\psi(\varepsilon) = \varepsilon$,  
 \item if $M'(s'')[w]\in \Delta^{(m)}$ then $M(s')[\psi(w)]=M'(s'')[w]$ and $\psi(wj)=\psi(w)j$ for $j\in [m]$, and 
 \item if $M'(s'')[w]=\<(q, I), x>\in \<Q', \{x\}>$ then $M(s')[\psi(w)]=\<q,x>$ and $\psi(wj)=\psi(w)I(j)$ for every $j\in [|I|]$. 
\end{itemize}
Then, the following property trivially holds. 
\begin{claim}\label{cl:del-prop1-general}
$M'(s'')/w = M(s')/\psi(w)\Theta_{p_0}$ for every $w\in V(M'(s''))$. 
\end{claim}

Let $w_1, w_2\in V_{\<Q', \{x\}>}(M'(s''))$ such that $M'(s'')[w_1]=M'(s'')[w_2]$, and let $\<(q,I), x>=M'(s'')[w_1]$. 
It follows from the third property of $\psi$ that $M(s')[\psi(w_1)]=M(s')[\psi(w_2)]=\<q,x>\in V_{\<Q,\{x\}>}(M(s'))$. 
Since $M$ has the dynamic FV property on $L$, 
$M(s')/\psi(w_1)j[\![s/u]\!]=M(s')/\psi(w_1)j[\![s/u]\!]$ for every $j\in [\rank_M(q)]$. 
By the first claim 
we get that $M'(s'')/w_1j'[\![B(s)/u]\!]=M'(s'')/w_2j'[\![B(s)/u]\!]$ for every $j'\in [|I|]$. 
Hence, $M'$ has the dynamic FV property on $B(L)$. 
\QED
\end{proof}

Next, we give a construction of the nonerasing normal form of an $\MTT$ in the style of this paper according to the proof of Lemma~7.11 of \cite{EM99}. 
Let $M=(Q,\Sigma,\Delta,q_0,R)$ be a nondeleting $\MTT$.  
For $p_1,\ldots,p_k\subseteq Q^{(1)}$, let $\Theta_{p_1\ldots p_k}=[\![\<q',x_i>\leftarrow y_1\mid \<q',x_i>\in \<Q,X_k>, q'\in p_i]\!]$. 
We first define a $\BREL$ $B=(2^{Q^{(1)}},\Sigma,\Sigma_M, 2^{Q^{(1)}}, R_B)$
where $\Sigma_M=\{[\sigma, p_1\ldots p_k]^{(k)}\mid \sigma\in \Sigma^{(k)}$, $k\geq 0$, $p_1,\ldots,p_k\subseteq Q^{(1)}\}$. 
For $\sigma\in \Sigma^{(k)}$ and $p_1,\ldots,p_k\subseteq Q^{(1)}$, 
let the rule 
\[
\sigma(p_1(x_1),\ldots,p_k(x_k))\to p([\sigma,p_1\ldots p_k](x_1,\ldots,x_k))
\]
be in $R_B$ where $p=\{q\in Q\mid \rhs_M(q,\sigma)\Theta_{p_1\ldots p_k}=y_1\}$. 
We construct $M'=(Q, \Sigma_M, \Delta, q_0, R')$. 
For $q\in Q^{(m)}$, $[\sigma, p_1\ldots p_k]\in \Sigma_M^{(k)}$, and $k, m\geq 0$, let the rule
\[
\<q, [\sigma, p_1\ldots p_k](x_1,\ldots, x_k)>(y_1,\ldots, y_m)\to \zeta
\]
be in $R'$ where $\zeta=\rhs_M(q,\sigma)\Theta_{p_1\ldots p_k}$ if $\rhs_M(q,\sigma)\Theta_{p_1\ldots p_k}\neq y_1$, 
and otherwise $\zeta=\bot(y_1)$ with  $\bot\in \Delta^{(1)}$. 

\begin{lemma}\label{lem:NonErasingPrvDFV}
\rm
Let $L\subseteq T_\Sigma$ and let $M=(Q,\Sigma,\Delta,q_0,R)$ be a nondeleting $\MTT$ that has the dynamic FV property on $L$. 
Let $B\com M'$ be the nondeleting and nonerasing $\MTTR$ obtained from $M$ by the construction given above. 
Then $M'$ has the dynamic FV property on $B(L)$. 
\end{lemma}
\begin{proof} From the construction, we have the following properties. 
\begin{enumerate}
\item \label{cl:ers-prop1}
$M'(s'') = M(s')[\![\<q, x>\leftarrow y_1 \mid M_q(s/u)=y_1]\!]$. 
\item \label{cl:ers-prop2}
$M'_q(B(s)/u)=M_q(s/u)$ for every $q\in Q$ such that $M_q(s/u)\neq y_1$. 
\end{enumerate}

Similar to the discussion in the last part of the proof of Lemma~\ref{lem:NonDelPrvDFV} using a correspondence between $V(M'(s''))$ and $V(M(s'))$, 
we obtain the lemma by the above properties. 
\QED
\end{proof}

It follows from Lemmas~\ref{lem:NonDelPrvDFV} and \ref{lem:NonErasingPrvDFV} that for every $\MTTRDFV$ (resp. $\MTTUDFV$) 
there exists a nondeleting and nonerasing $\MTTRDFV$ (resp. $\MTTUDFV$) equivalent with it.

Next we show that the tree translation realized by any $\MTTUDFV$ can be expressed by an $\ATTU$.  
\begin{lemma}\label{lem:MTTDFV-TREL+ATT}
\rm
$\CMTTUDFV\subseteq \CATTU$. 
\end{lemma}

%
Henceforth, we assume w.l.o.g. that an $\MTTUDFV$ is nondeleting and nonerasing. 
We define the notion of reachable state.
\begin{definition}
\rm
Let $Q'\subseteq Q$ and $s\in T_\Sigma(X)$. 
For $u\in V(s)$, the \emph{set of reachable states of $s$ at $u$} from $Q'$, denoted by $\Sts_M(Q',s,u)$, 
is the set $\bigcup_{q'\in Q'} \{q\in Q\mid \exists v\in V(\xi_{q'}).~\xi_{q'}[v]=\<q, x>\}$
where $\xi_{q'}=M_{q'}(s[u\leftarrow x])$. 
%
\end{definition}
The idea for proving that for every $\MTTDFV$ there exists an equivalent $\ATTU$ is as follows. 
The dynamic FV property demands (semantic) equivalence of parameter trees
only for those states that are processing a given input node simultaneously. 
Therefore, we use the top-down relabeling to add to each input node $u$ the information
which states are processing $u$, i.e., the set $\Sts_M(\{q_0\},s,u)$, where
$s$ is the input tree. Then, we construct for regular look-around a composition 
of the bottom-up relabeling of the given $\MTTRDFV$ and the top-down relabeling. 
Further, the constructed $\ATT$ has a synthesized attribute for each state $q$ of the $\MTT$, and
an inherited attribute for each pair $\langle q,j\rangle$, where $j\in[\rank_Q(q)]$.
Last but not least, of those parameter trees which are known to be equivalent, we pick one 
and construct from its ``top-part'' a rule for the corresponding inherited attribute. 
The top-part is defined similarly to the definition of $\Top$ in Definition~\ref{def:consistent},
except that if we are computing the top-part of a parameter tree of a state appearing in $\rhs_M(q,\sigma)$, then an occurrence
of the parameter $y_j$ is now replaced by the inherited attribute $\langle q,j\rangle$.
%
The construction from $\MTTU$s with the dynamic FV to $\ATTU$ following the idea is given in Definition~\ref{def:constructAttDFV}. 

To illustrate our idea, consider the above example MTT $M_{\text{dyn}}$. 
We construct an ATT that has a rule set $R_{[\sigma,Q']}$ for each label $\sigma$ and subset $Q'$ 
of states reachable from $\{q_0\}$ (see Figure~\ref{fig:ex-dfv-att}). 
Here we focus on $Q'=\{q_0\}$ and $\sigma=a$. 
For the synthesized attribute $q_0$ in $Q'$, we construct its rule as the top part of $\rhs_M(q_0,a)$:
$q_0(\pi) \to f(q_1(\pi 1), q_1(\pi 1))$, 
where state calls $\<q, x_1>$ were replaced by $q(\pi 1)$. 
For the inherited attribute $\<q_1, 1>$ which corresponds to the first parameter of $q_1$, 
we construct the rhs of $\<q_1,1>(\pi 1)$ by choosing a candidate from the first arguments of $\<q_1,x_1>$ 
appearing in $\rhs_M(q', \sigma)$ for $q'\in Q'$. Here we choose $\<q_2,x_1>$ and 
construct the rule
$\<q_1,1>(\pi 1)\to q_2(\pi 1)$, 
where $\<q_2,x_1>$ was replaced by the synthesized attribute $q_2(\pi 1)$. 
The non-dummy rules of the ATT are shown in Figure~\ref{fig:ex-dfv-att}.
%
\begin{figure}[t]
\centering
\[
\begin{array}{rlcl}
R_{[a,\{q_0\}]}:  
& q_0(\pi) 			&\to& f(q_1(\pi 1), q_1(\pi 1))\\
& \<q_1,1>(\pi 1)	&\to& q_2(\pi 1)\\
& \<q_3,1>(\pi 1) 	&\to& e\\
R_{[a, \{q_1,q_2,q_3\}]}:
& q_1(\pi) 			&\to& a(q_1(\pi 1))\\
& q_2(\pi) 			&\to& a(a(q_2(\pi 1)))\\
& q_3(\pi) 			&\to& a(q_3(\pi 1))\\
& \<q_1,1>(\pi 1)	&\to& b(\<q_1,1>(\pi))\\
& \<q_3,1>(\pi 1)   &\to& a(\<q_3,1>(\pi))\\
R_{[e,\{q_0\}]}:
& q_0(\pi)       	 	&\to& e\\
R_{[e, \{q_1,q_2,q_3\}]}:
& q_1(\pi)			&\to& \<q_1,1>(\pi 1)\\
& q_2(\pi) 			&\to& e\\
& q_3(\pi) 			&\to& \<q_3,1>(\pi 1)
\end{array}
\]
\caption{The ATT constructed from the MTT $M_{\text{dyn}}$ with the dynamic FV property}
\label{fig:ex-dfv-att}
\end{figure}

\begin{definition}\label{def:constructAttDFV}
\rm
Let $H=(P,\Sigma,\Gamma,F,R_H)$ and $M=(Q,\Gamma,\Delta,q_0,R)$ be a $\TRREL$ and an $\MTT$ such that 
$H\com M$ has the dynamic FV property. 
Let $\bot$ be an arbitrary symbol in $\Delta^{(0)}$. 

\begin{enumerate}
\item
The top-down relabeling $E(M)=(2^Q,\Gamma,\Gamma_M,\{q_0\},R_E)$ is defined as follows. 
 Let $\Gamma_M=\Gamma\times 2^Q$. 
 For every $Q'\subseteq Q$ and $\sigma\in \Gamma^{(k)}$, 
 let $\sigma'=[\sigma, Q']$ be in $\Gamma_M^{(k)}$ and
 let
$\<Q', \sigma(x_1,\ldots,x_k)>\to \sigma'(\<Q_1,x_1>,\ldots,\<Q_k,x_k>)$ 
 be a rule in $R_E$, where $Q_i=\Sts_M(Q', \sigma(x_1,\ldots,x_k), i)$ for $i\in [k]$. 
It should be clear that for every input tree $s\in T_\Gamma$ and node $u\in V(s)$, 
$D(s)[u]=[s/u, \Sts_M(\{q_0\},s,u)]$. 

\item
The ATT $A(M)=(S, I, \Gamma_M, \Delta, q_0, R')$ is defined as follows. 
$S=Q$, $I=\{\langle q, j\rangle \mid q\in Q^{(m)}, m\geq 0, j\in [m]\}$, and
$R'=\bigcup_{\sigma'\in \Gamma_M} R'_{\sigma'}$ where $R'_{\sigma'}$ is defined below.  
Let $\sigma'=[\sigma,Q']\in\Gamma_M$ and $k=\rank_\Sigma(\sigma)$. 
For every $i\in[k]$, $q\in Q$, and $j\in[\rank_Q(q)]$, 
we fix trees $\zeta^{\sigma'}_{i,q,j}$ that will be used to define the $\rhs_A(\sigma',\<q, j>(\pi i))$.  
---
Let $F=\emptyset$.
We fix an order $q_1,q_2,\dots,q_n$ on the states in $Q'$.
We now traverse, starting with $\nu=1$ the trees $\zeta=\rhs_M(q_\nu,\sigma)$ in \emph{post-order}. 
Whenever a state call $\<q,x_i>$ (at node $v$ of $\zeta$) is
encountered and $\<q,x_i>\not\in F$, then we change $F$ to $F\cup\{\<q,x_i>\}$ and
we define $\zeta^{\sigma'}_{i,q,j}=\Top'_{q_\nu}(\zeta/vj)$ for all $j\in[\rank_Q(q)]$. 
The function 
$\Top'_{q}:T_{\Delta\cup \langle Q, X\rangle}(Y_m)\to T_\Delta(\{\alpha(\pi i)\mid\alpha\in S, 
i\geq 0\}\cup\{\beta(\pi)\mid\beta\in I\})$ for $q\in Q^{(m)}$ is defined as:  
$\Top'_q(\zeta)$ is obtained from $\Top(\zeta)$ by replacing $y_j\in Y_m$ with $\<q,j>$. 
%
For every $\sigma'=[\sigma,Q']\in \Gamma_M^{(k)}$ with $k\geq 0$, 
we construct the set $R'_{\sigma'}$ in the following way:
\begin{itemize}
\item
For $q\in Q'$, let the rule $q(\pi) \to \Top'_{q}(\rhs_M(q,\sigma))$ be in $R'_{\sigma'}$. 
For $q\in Q-Q'$, let the rule $q(\pi) \to \bot$ be in $R'_{\sigma'}$. 

\item 
For every $i\in [k]$, $q\in Q^{(m)}$, and $j\in [\rank_Q(q)]$, 
if $\zeta^{\sigma'}_{i,q,j}$ is defined then we add $\<q, j>(\pi i)\to \zeta^{\sigma'}_{i,q,j}$ to $R'_{\sigma'}$;
otherwise add $\<q, j>(\pi i)\to \bot$ to $R'_{\sigma'}$. 
\end{itemize}
\end{enumerate}
\end{definition}

Henceforth, we assume that  $\MTTU$ $H\com M$ has the dynamic FV property where and $M$ is non-deleting and 
non-erasing. Let $E=E(M)$, $E_A=H\com E$, and $A=A(M)$. To prove Lemma~\ref{lem:MTTDFV-TREL+ATT}, 
we show the non-circularity of $A$ on $E_A(T_\Sigma)$ (Lemma~\ref{lem:non-circularity}) and then 
the equivalence of $H\com M$ and $E_A\com A$ (Lemma~\ref{lem:equivalence}). 
Note that the ATT $A$ may not be non-circular on $T_\Gamma-E_A(T_\Sigma)$. 
However, by Theorem~15 of \cite{BE00}, we can construct a $\TRREL$ $E'$ and an $\ATT$ $A'$ 
such that $A'$ is non-circular on all input trees and $\tau_{E'}\com \tau_{A'}=\tau_{E_A}\com \tau_{A}$. 
Thus,  $\CMTTUDFV\subseteq \CATTU$ follows the two lemmas. 

Let $H=(P,\Sigma,\Gamma,F,R_H)$, $M=(Q,\Gamma,\Delta,q_0,R)$, and $[\![s]\!]=[\![\<r,x> \leftarrow M_r(s)\mid r\in Q]\!]$ for $s\in T_\Gamma$.
Let $A=(S, I, \Gamma_M, \Delta, q_0, R')$. 
First, we give a proof of the non-circularity of $A$ on the image of $E_A$. 
To prove it, we use the following lemma. 
\begin{lemma}\label{lem:size-increase2}
\rm
The following properties hold. 
\begin{enumerate}
\item For every $s\in T_\Gamma$, $s_0\in H(T_\Sigma)$, $u\in V(s_0)$ such that $s_0/u=s$, 
$\langle q_1,j\rangle \in I$, and $q_2\in S$, 
if there exists a path from $(\<q_1,j>,\varepsilon)$ to $(q_2,\varepsilon)$ in $D_{A}(E_{Q'}(s))$ 
where $Q'=\Sts_M(\{q_0\},s_0,u)$, 
then $\ct_M(s_0,u,q_1)\neq \emptyset$ and $\size(t_1/j[\![s]\!])<\size(t_2[\![s]\!])$ for every 
$t_1\in \ct_M(s_0,u,q_1)$ and $t_2\in \ct_M(s_0,u,q_2)$. 
\item For every $s\in T_\Gamma$, $s_0\in H(T_\Sigma)$, $u\in V(s_0)$ 
such that $s_0/u=s$ and $s_0[u]=\sigma \in \Gamma^{(k)}$ with $k>0$,
$\<q_1,j> \in I$, $a\in S\cup I$, and $i, i'\in [k]$,
if there exists a path from $(\<q_1, j>,i)$ to $(a,i')$ in $D_{A}(E_{Q'}(s))$ 
where $Q'=\Sts_M(\{q_0\},s_0,u)$, then
$\ct_M(s_0,ui,q_1)\neq \emptyset$ and
\begin{itemize}
 \item if $a=q_2\in S$ then $\size(t_1/j[\![s/i]\!])<\size(t_2[\![s/i']\!])$
          for every $t_1\in \ct_M(s_0,ui,q_1)$ and $t_2\in \ct_M(s_0,ui',q_2)$. 
 \item if $a=\<q_2, j'>\in I$ then $\size(t_1/j[\![s/i]\!])<\size(t_2/j'[\![s/i']\!])$
          for every $t_1\in \ct_M(s_0,ui,q_1)$ and $t_2\in \ct_M(s_0,ui',q_2)$. 
\end{itemize}
\end{enumerate}
\end{lemma}
\begin{proof}
We prove Statements 1 and 2 by two (nested) inductions on the structure of $s$ as follows:
\begin{itemize}
 \item Base case: prove that Statement 1 holds for $s\in \Gamma^{(0)}$. Statement 2 holds for $s$ because $k=0$. 

 \item Induction step I: by using the induction hypothesis of Statement 1 on the subtrees of $s$, 
         prove that Statement 2 holds for $s\in T_\Sigma$ by induction on the length of the path in       
         $D_{A}(E_{Q'}(s))$.
 \item Induction step II: by using the fact that Statement 2 holds for $s\in T_\Gamma$, 
          prove that Statement 1 holds for $s$.   
\end{itemize}

\smallskip\noindent
\emph{Base case.} 
Let $s=\sigma\in \Gamma^{(0)}$. Let $s_0\in H(T_\Sigma)$ and $u\in V(s_0)$ such that $s_0/u=s$.  
Let $Q'=\Sts_M(\{q_0\},s_0,u)$, $\<q_1, j>\in I$, and $q_2\in S$.  
Assume that there exists a path from $(\<q_1,j>,\varepsilon)$ to $(q_2,\varepsilon)$ in 
$D_{A}(E_{Q'}(\sigma))$. 
Since $E_{Q'}(\sigma)=[\sigma,Q']\in \Gamma_M^{(0)}$,  
the path is only a direct edge from $(\<q_1, j>,\varepsilon)$ to $(q_2,\varepsilon)$ 
in $D_{A}([\sigma,Q'])$. 
From the construction of rules for labels with rank 0, 
the edge originates from the rule $q_2(\pi)\to \Top'_{q_2}(\rhs_M(q_2,\sigma))$ of $A$, and $q_2\in Q'$. 
The edge from $(\<q_1,j>,\varepsilon)$ to $(q_2,\varepsilon)$ exists only if 
$\<q_1, j>(\pi)$ occurs in the rhs of the rule of $A$. 
By the definition of $\Top'_{q_2}$, inherited attributes with states other than $q_2$ do not occur 
in the rhs of the rule.  Hence, $q_1$ must be equal to $q_2$. 
Since $q_2\in Q'$, we get $\ct_M(s_0,u,q_1)=\ct_M(s_0,u,q_2)\neq \emptyset$. 
Let $t\in \ct_M(s_0,u,q_2)$. 
Since $t[\varepsilon]=\<q_2, x>$, 
$t[\![s]\!]=\rhs_M(q_2,\sigma)[y_l\leftarrow t/l[\![s]\!]\mid l\in [\rank_Q(q_2)]]$.  
Since $M$ is nonerasing and nondeleting, $\size(t/j[\![s]\!])<\size(t[\![s]\!])$. 
By the dynamic FV property of $M$, we get $\size(t_1/j[\![s]\!])<\size(t_2[\![s]\!])$  
for every $t_1,t_2\in \ct_M(s_0,u,q_1)$. 

\bigskip\noindent
\emph{Induction step I.}
Let $s=\sigma(s_1,\ldots,s_k)$ and $k>0$. Let $s_0\in H(T_\Sigma)$ and $u\in V(s_0)$ such that $s_0/u=s$.  
Let $Q'=\Sts_M(\{q_0\},s_0,u)$, $\<q_1, j>\in I$, and $a\in S\cup I$. 
Let $i, i'\in [k]$. 
Assume that there exists a path from $(\<q_1,j>,i)$ to $(a,i')$ in $D_{A}(E_{Q'}(s))$. 
By the definition of $E$, $E_{Q'}(s)[\varepsilon]=[\sigma,Q']$ and 
$E_{Q'}(s)/i=E_{Q_i}(s_i)$ for $i\in [k]$  
where $Q_i=\Sts_M(Q', s, i)$. 

We can regard the path $w$ from $(\<q_1,j>,i)$ to a node $v\in (S\cup I)\times [k]$ in 
$D_{A}(E_{Q'}(s))$ as a sequence of nodes of $D_{A}([\sigma,Q'])$ such that every two consecutive nodes 
are connected by an edge of $D_{A}([\sigma,Q'])$, or a path of a subgraph 
$D_{A}(E_{Q_i}(s_i))$ for some $i\in [k]$. 
We denote by $\to_\varepsilon$ and $\to_i$ connections by an edge of $D_{A}([\sigma,Q'])$ 
and a path of a subgraph $D_{A}(E_{Q_i}(s_i))$, respectively. 
Then, we prove by induction on the length of $w$, by using the induction hypothesis of Statement 1 on the subtrees of $s$, that for every node $v$ in $(S\cup I)\times [k]$ on $w$, 
\begin{itemize}
 \item if $v=(\<q,j'>,i')$ where $\<q,j'>\in I$ and $i'\in [k]$ then 
          $\size(t_1/j[\![s_i]\!])<\size(t_2/j'[\![s_{i'}]\!])$ for every $t_1\in \ct_M( s_0, ui, q_1)$ and $t_2\in \ct_M( s_0, ui', q)$.  
 \item if $v=(q,i')$ where $q\in S$ and $i'\in [k]$ then 
          $\size(t_1/j[\![s_i]\!])<\size(t_2[\![s_{i'}]\!])$ for every $t_1\in \ct_M( s_0, ui, q_1)$ and $t_2\in \ct_M( s_0, ui', q)$.  
\end{itemize}

\smallskip\noindent
\emph{Base case IB. }
Let $w=(\<q_1,j>,i)\to_i (q,i)$. 
There is a path from $(\<q_1,j>,\varepsilon)$ to $(q,\varepsilon)$ in $D_{A}(E_{Q_i}(s_i))$. 
By the induction hypothesis of Statement 1, we have 
$\ct_M(s_0,ui,q_1)\neq \emptyset$ and 
$\size(t_1/j[\![s_i]\!])<\size(t_2[\![s_i]\!])$ 
for every $t_1\in \ct_M(s_0,ui,q_1)$ and $t_2\in \ct_M(s_0,ui,q)$. 

\smallskip\noindent
\emph{Induction step I1.} 
Let $w=(\<q_1,j>,i)\to^+ (\<q',j'>,i') \to_{i'} (q,i')$. 
Since $(\<q_1,j>,i)\to^+ (\<q',j'>,i')$, 
by induction hypothesis of the inner induction, 
we have $\ct_M(s_0,ui,q_1)\neq \emptyset$ and
$\size(t_1/j[\![s_i]\!])<\size(t_2/j'[\![s_{i'}]\!])$ 
for every $t_1\in \ct_M(s_0,ui,q_1)$ and $t_2\in \ct_M(s_0,ui',q')$. 
Since $(\<q',j'>,i') \to_{i'} (q,i')$, there is a path from 
$(\<q',j'>,\varepsilon)$ to $(q,\varepsilon)$ in $D_{A}(E_{Q_{i'}}(s_{i'}))$. 
By induction hypothesis of Statement 1, 
$\ct_M(s_0,ui',q')\neq \emptyset$ and 
$\size(t_2/{j'}[\![s_{i'}]\!])<\size(t_3[\![s_{i'}]\!])$ 
for every $t_2\in \ct_M(s_0,ui',q')$ and $t_3\in \ct_M(s_0,ui',q)$. 
Thus, $\size(t_1/j[\![s_i]\!])<\size(t_3[\![s_{i'}]\!])$ 
for every $t_1\in \ct_M(s_0,ui,q_1)$ and $t_3\in \ct_M(s_0,ui',q)$. 

\smallskip\noindent
\emph{Induction step I2.} 
Let $w=(\<q_1,j>,i)\to^+ (q',i'') \to_{\varepsilon} (\<q,j'>,i')$. 
From the construction of rules for the symbol $[\sigma,Q']$, 
the edge from $(q',i'')$ to $(\<q,j'>,i')$ originates from the rule 
$\<q,j'>(\pi i')\to \Top'_{q''}(\zeta/uj')$ 
such that $q''\in Q'$, $\zeta=\rhs_M(q'', \sigma)$, 
$\zeta[v]=\<q, x_{i'}>$ for some $v\in V(\zeta)$, 
and $q'(\pi i'')$ occurs in $\Top'_{q''}(\zeta/uj')$. 
By the definition of $\Top'_{q''}$, $\<q', x_{i''}>$ occurs in $\zeta/vj'$ at some node $v'$ such that 
every ancestor of $v'$ has a symbol in $\Delta$. 
Since $q''\in Q'$, $\ct_M(s_0,u,q'')\neq \emptyset$. 
Let $t\in \ct_M(s_0,u,q'')$. 
Let $[\![rhs]\!] = [\![\<r, p> \leftarrow \rhs_M(r, \sigma) \mid r\in Q]\!]$. 
For $l\in \{i',i''\}$ let $[\![\setminus l]\!] = [\![\<r,x_c> \leftarrow M_r(s/c) \mid r\in Q, c\in [k]-\{l\}]\!]$ and $[\![l]\!] = [\![\<r,x_l> \leftarrow \<r,x> \mid r\in Q]\!]$, and let $[\![s, l]\!]=[\![rhs]\!][\![\setminus l]\!][\![l]\!]$. 
Let $\eta_{i'}=t[\![s, i']\!]$ and then $\eta_{i'}=\zeta\psi[\![\setminus i']\!][\![i']\!]$ where 
$\psi=[y_l\leftarrow t/l[\![rhs]\!]\mid l\in [\rank_Q(q'')]]$. 
Then, there exists $\tilde{v}\in V(\eta_{i'})$ such that $\eta_{i'}/\tilde{v}=(\zeta/v)\psi[\![\setminus i']\!][\![i']\!]$. 
Since $M$ is nondeleting, $\eta_{i'}$ is a subtree of $M(s_0[ui'\leftarrow x])$. 
Thus, $\eta_{i'}/\tilde{v}\in \ct_M(s_0, ui', q) \neq \emptyset$. 
Let $\eta_{i''}=t[\![rhs]\!][\![\setminus i'']\!][\![i'']\!]$. By the same argument, 
$\eta_{i''}=\zeta\psi[\![\setminus i'']\!][\![i'']\!]$ and 
there exists $\tilde{v}'\in V(\eta_{i''})$ such that $\eta_{i''}/\tilde{v}'=(\zeta/vj'v')\psi[\![\setminus i'']\!][\![i'']\!]$. 
Also, $\eta_{i''}/\tilde{v}'\in \ct_M(s_0, ui'', q') \neq \emptyset$. 
Since $[\![\setminus i']\!][\![i']\!][\![s_{i'}]\!]=[\![\setminus i'']\!][\![i'']\!][\![s_{i''}]\!]$ and $M$ is nondeleting, 
$\eta_{i''}/\tilde{v}'[\![s_{i''}]\!]$ is a subtree of $\eta_{i'}/\tilde{v}j'[\![s_{i'}]\!]$. 
Thus, $\size(\eta_{i''}/\tilde{v}'[\![s_{i''}]\!])\leq \size(\eta_{i'}/\tilde{v}j'[\![s_{i'}]\!])$. 
By the dynamic FV property, $\size(t_2[\![s_{i''}]\!])\leq \size(t_3/j'[\![s_{i'}]\!])$ for every $t_2\in \ct_M(s_0,ui'',q')$ and $t_3\in \ct_M(s_0,ui',q)$. 
Since $(\langle q_1,j\rangle,i)\to^+ (q',i'')$, 
by induction hypothesis of the inner induction, 
$\ct_M(s_0,ui, q_1)\neq \emptyset$ and 
$\size(t_1/j[\![s_i]\!])<\size(t_2[\![s_{i''}]\!])$ for every $t_1\in \ct_M(s_0,ui,q_1)$ 
and $t_2\in \ct_M(s_0,ui'',q')$. 
Hence, for every $t_1\in \ct_M(s_0,ui,q_1)$ and $t_3\in \ct_M(s_0,ui',q)$, 
$\size(t_1/j[\![s_{i}]\!])<\size(t_3/j'[\![s_{i'}]\!])$.

\bigskip\noindent
\emph{Induction step II.}
Let $s=\sigma(s_1,\ldots,s_k)$ and $k>0$. Let $s_0\in B(T_\Sigma)$ and $u\in V(s_0)$ such that $s_0/u=s$. 
Let $Q'=\Sts_M(\{q_0\},s,u)$, $\<q_1, j>\in I$, and $q_2\in S$. 
Assume that there exists a path $w$ from $(\<q_1,j>,\varepsilon)$ to $(q_2,\varepsilon)$ in $D_{A}(E_{Q'}(s))$. 
By the definition of $E$, $E_{Q'}(s)[\varepsilon]=[\sigma,Q']$ and 
$E_{Q'}(s)/i=E_{Q_i}(s_i)$ for $i\in [k]$  
where $Q_i=\Sts_M(Q', s, i)$. 

\smallskip
\emph{Case 1.}
Let $w=(\<q_1,j>,\varepsilon)\to_\varepsilon (q_2,\varepsilon)$. 
By the same argument of the base case, 
we can show that $q_1=q_2$, $\ct_M(s_0,u,q_1)\neq \emptyset$, and 
$\size(t_1/j[\![s]\!])<\size(t_2[\![s]\!])$ for every $t_1,t_2\in \ct_M(s_0,u,q_1)$. 

\smallskip
\emph{Case 2.}
Let $w=(\<q_1,j>,\varepsilon)\to_\varepsilon (\<q', j'>,i) \to^+ (q'',i') \to_{\varepsilon} (q_2,\varepsilon)$. 
It is sufficient to show that 
\begin{enumerate}[(i)]
\item $\ct_M(s_0,u,q_1)\neq \emptyset$ and $\size(t_1/j[\![s]\!])\leq \size(t_2/j'[\![s_i]\!])$ for every call trees $t_1\in \ct_M(s_0,u,q_1)$ and $t_2\in \ct_M(s_0,ui,q')$, 
\item $\ct_M(s_0,ui,q')\neq \emptyset$ and $\size(t_2/j'[\![s_i]\!])<\size(t_3[\![s_{i'}]\!])$ for every call trees $t_2\in \ct_M(s_0,ui,q')$ and $t_3\in \ct_M(s_0,ui',q'')$, and
\item $\ct_M(s_0,ui',q'')\neq \emptyset$ and $\size(t_3[\![s_{i'}]\!])\leq \size(t_4[\![s]\!])$ for every call trees $t_3\in \ct_M(s_0,ui',q'')$ and $t_4\in \ct_M(s_0,u,q_2)$. 
\end{enumerate}
Since (ii) can be shown by the induction hypothesis of Statement 2, we will show (i) and (iii) below. 
Let $[\![rhs]\!] = [\<r, x> \leftarrow \rhs_M(r, \sigma)\mid r\in Q]$, and 
for $l\in \{i,i'\}$ let $[\![\setminus l]\!] = [\![\<r,x_c> \leftarrow M_r(s/c) \mid r\in Q, c\in [k]-\{l\}]\!]$, 
$[\![l]\!] = [\![\<r,x_l> \leftarrow \<r, x> \mid r\in Q]\!]$, and $[\![s, l]\!]=[\![rhs]\!][\![\setminus l]\!][\![l]\!]$. 
Note that the above substitutions are nondeleting because $M$ is nondeleting.

(i) Since $(\<q_1,j>,\varepsilon)\to_\varepsilon (\<q', j'>,i)$, 
from the construction of rules for the symbol $[\sigma,Q']$, 
the edge originates from the rule $\<q',j'>(\pi i)\to \Top'_{q_1}(\zeta/vj')$ 
such that $q_1\in Q'$, $\zeta=\rhs_M(q_1, \sigma)$, 
$\zeta[v]=\<q', x_i>$ for some $v\in V(\zeta)$, 
and $\<q_1,j>(\pi)$ occurs in $\Top'_{q_1}(\zeta/vj')$. 
By the definition of $\Top'_{q_1}$, $y_j$ occurs in $\zeta/vj'$ at some node $v'$. 
Since $q_1\in Q'$, $\ct_M(s_0,u,q_1)\neq \emptyset$.  
Let $t_1\in \ct_M(s_0,u,q_1)$ and $\eta=t_1[\![s,i]\!]$. 
Since $t_1[\varepsilon]=\<q_1, x>$, $\eta=\zeta\psi[\![\setminus i]\!][\![i]\!]$ where $\psi=[y_l\leftarrow t_1/l[\![rhs]\!]\mid l\in [\rank_Q(q_1)]]$. 
There exists $\tilde{v}\in V(\eta)$ such that 
$\eta/\tilde{v}=(\zeta/v)\psi[\![\setminus i]\!][\![i]\!]$ and
$\eta/\tilde{v}j'v'=(\zeta/vj'v')\psi[\![\setminus i]\!][\![i]\!]=t_1/j[\![rhs]\!][\![\setminus i]\!][\![i]\!]$ because $y_j=\zeta[vj'v']$. 
Then, $\eta[\tilde{v}]=\<q', x>$.  Thus, $\eta/\tilde{v}\in \ct_M(s_0,ui,q')\neq \emptyset$. 
Let $t_2=\eta/\tilde{v}$. Since $t_2/j'v'=t_1/j[\![s,i]\!]$, $t_2/j'v'[\![s_i]\!]=t_1/j[\![s,i]\!][\![s_i]\!]=t_1/j[\![s]\!]$. 
Thus, $\size(t_1/j[\![s]\!])=\size(t_2/j'v'[\![s_i]\!])\leq \size(t_2/j'[\![s_i]\!])$. 
By the dynamic FV property of $M$, (i) holds. 

(iii) Since $(q'',i') \to_{\varepsilon} (q_2,\varepsilon)$, 
from the construction, the edge originates from the rule 
$q_2(\pi)\to \Top'_{q_2}(\zeta)$ where $\zeta=\rhs_M(q_2,\sigma)$, and 
$q_2\in Q'$. 
By the definition of $\Top'_{q_2}$, $\zeta[v]=\<q'', x_{i'}>$ for some $v\in V(\zeta)$ and 
every ancestor of $v$ has a symbol in $\Delta$.  
Since $q_2\in Q'$, $\ct_M(s_0,u,q_2)\neq \emptyset$. 
Let $t_3\in \ct_M(s_0,u,q_2)$ and $\eta=t_3[\![s,i']\!]$. 
Since $t_3[\varepsilon]=\<q_2, x>$, $\eta=\zeta\psi[\![\setminus i']\!][\![i']\!]$ where 
$\psi=[y_l\leftarrow t_3/l[\![rhs]\!]\mid l\in [\rank_Q(q_2)]]$. 
There exists $\tilde{v}\in V(\eta)$ such that 
$\eta/\tilde{v}=(\zeta/v)\psi[\![\setminus i']\!][\![i']\!]$. 
Then, $\eta[\tilde{v}]=\<q'', x>$.  Thus, $\eta/\tilde{v}\in \ct_M(s_0,ui',q'')\neq \emptyset$. 
Since $\eta[\![s_{i'}]\!]=t_3[\![s,i']\!][\![s_{i'}]\!]=t_3[\![s]\!]$, 
$\size(\eta/\tilde{v}[\![s_{i'}]\!])\leq \size(\eta[\![s_{i'}]\!])=\size(t_3[\![s]\!])$. 
By the dynamic FV property of $M$, (iii) holds. 
\QED
\end{proof}

\begin{lemma}\label{lem:non-circularity}
\rm
$A$ is non-circular on $E_A(T_\Sigma)$. 
\end{lemma}
\begin{proof}
The proof is done by contradiction. Assume that $A$ is circular on $E_A(T_\Sigma)=E(H(T_\Sigma))$. 
Then there exists a tree $s\in H(T_\Sigma)$ such that $D_{A}(E(s))$ has a cycle. 
Let $u\in V(s)$ be a node such that $E(s)/u$ is a minimal subtree that includes the cycle. 
Then $s/u\in T_\Gamma-\Gamma^{(0)}$ because no cycle can be made in $D_{A}(\sigma')$ for any $\sigma'\in \Gamma^{(0)}_M$. 
Let $\sigma=s[u]\in \Gamma^{(k)}$ with $k>0$.  
Let $Q'=\Sts_M(\{q_0\},s,u)$. 
By the definition of $E$ and $A$, $Q'\neq \emptyset$ 
because if $Q'=\emptyset$ then there were no edges between any two nodes on $D_{A}(E(s))/u$. 
Then, there is a cycle of length greater than 0 from some $(\langle q, j\rangle,i)$ to itself in $D_{A}(E_{Q'}(s/u))$ for some $\langle q, j\rangle\in I$ and $i\in [k]$. 
By Lemma~\ref{lem:size-increase2}, $\size(t/j[\![s/ui]\!])<\size(t/j[\![s/ui]\!])$ for any $t\in \ct_M(s,ui,q)$. This is a contradiction. 
\QED
\end{proof}

We show the equivalence of $H\com M$ and $E_A\com A$.  
\begin{lemma}\label{lem:equivalence}
\rm
$\tau_H\com \tau_M=\tau_{H}\com \tau_{E}\com \tau_{A}$. 
\end{lemma}
\begin{proof}
For $s\in H(T_\Sigma)$, $u\in V(s)$, and $q\in Q'$ where $Q'=\Sts_M(\{q_0\}, s, u)$ , let
\begin{eqnarray*}
 {[s,u,q]_M} &=& [y_l\leftarrow t/l[\![s/u]\!]\mid t\in \ct_M(s,u,q), l\in [\rank_Q(q)]], \\
 {[s,u]_A} &=& [\<q,l>(\varepsilon) \leftarrow t/l[\![s/u]\!]\mid q\in Q', t\in \ct_M(s,u,q), l\in [\rank_Q(q)]].
\end{eqnarray*}
Note that by the dynamic FV property of $M$ on $H(T_\Sigma)$, $[s,u,q]_M$ and $[s,u]_A$ are well-defined.
We prove the following statements for every $s\in T_\Gamma$. 
\begin{enumerate}
\item For every $s_0\in H(T_\Sigma)$, $u\in V(s)$ such that $s_0/u=s$, 
and $q\in Q'$ where $Q'=\Sts_M(\{q_0\}, s_0, u)$,
\[
M_q(s)[s_0,u,q]_M=\nf(\Rightarrow_{A,E_{Q'}(s)}, q(\varepsilon))[s_0,u]_A. 
\]

\item For every $s_0\in H(T_\Sigma)$, $u\in V(s)$ such that 
$s_0/u=s$ and $s_0[u]=\sigma\in \Sigma^{(k)}$ with $k\geq 0$, 
$q\in Q'$ where $Q'=\Sts_M(\{q_0\},s_0,u)$, and $v\in V(\zeta_q)$, 
\[
(\zeta_q/v)\theta[s_0,u,q]_M=\nf(\Rightarrow_{A, E_{Q'}(s)}, \Top'_{q}(\zeta_q/v)[\pi\leftarrow \varepsilon])[s_0,u]_A
\]
where $\zeta_{q}=\rhs_M(q,\sigma)$ 
 and $\theta=[\![\<r, x_i> \leftarrow M_r(s/i) \mid r\in Q, i\in [k]]\!]$. 
\end{enumerate}

We first prove that Statement 2 implies Statement 1 for all $s\in T_\Gamma$. 

\smallskip\noindent
\emph{(2$\Longrightarrow$1).} 
Let $s=\sigma(s_1,\ldots,s_k)\in T_\Gamma$ where $\sigma\in \Gamma^{(k)}$, $k\geq 0$, and $s_1,\ldots,s_k\in T_\Gamma$. Let $s_0\in H(T_\Sigma)$ and $u\in V(s)$ such that $s_0/u=s$, and $s'=E_{Q'}(s)$. 
Let $q\in\Sts_M(\{q_0\}, s_0, u)$.  
Let $\xi=\rhs_M(q,\sigma)$ and $\theta=[\![\<r, x_i> \leftarrow M_r(s/i) \mid r\in Q, i\in [k]]\!]$. We assume that Statement 2 holds for $s$. 
\begin{eqnarray*}
&&M_q(s)[s_0,u,q]_M\\
&=&\xi\theta[s_0,u,q]_M\\
&=&\nf(\Rightarrow_{A,s'}, \Top'_{q}(\xi)[\pi\leftarrow \varepsilon])[s_0,u]_A~~\mbox{(by Statement 2 with $v=\varepsilon$)}\\
&=&\nf(\Rightarrow_{A,s'}, q(\pi)[\pi\leftarrow \varepsilon])[s_0,u]_A~~~~~\mbox{(by the construction)}.
\end{eqnarray*}

By using the above fact, we prove that Statement 2 holds for all $s\in T_\Gamma$ by induction on the structure of $s$. 

\smallskip\noindent
\emph{Base case.} 
Let $s=\sigma\in \Gamma^{(0)}$. Let $s_0\in H(T_\Sigma)$ and $u\in V(s_0)$ such that $s_0/u=s$. 
Let $Q'=\Sts_M(\{q_0\},s_0,u)$ and $q\in Q'$. Let $m=\rank_Q(q)$. 
Let $\zeta=\rhs_M(q,\sigma)$, and let $v\in V(\xi)$ and $\xi=\zeta/v$. 
Since $E_{Q'}(\sigma)=[\sigma,Q']=\sigma'\in \Gamma_M^{(0)}$ and $q\in Q'$, 
from the construction of rules for labels with rank 0, 
the edge originates from the rule $q(\pi)\to \Top'_{q}(\xi)$ of $A$. 
By the definition of $\Top'_{q}$ and $\xi\in T_{\Delta\cup Y_m}$, 
\begin{eqnarray*}
&&\nf(\Rightarrow_{A,\sigma'}, \Top'_{q}(\xi)[\pi\leftarrow \varepsilon])[s_0,u]_A\\
&=& \nf(\Rightarrow_{A,\sigma'}, \xi[y_l\leftarrow \<q,l>(\varepsilon)\mid l\in [m]])[s_0,u]_A\\
&=& \xi[y_l\leftarrow \<q,l>(\varepsilon)\mid l\in [m]][s_0,u]_A\\
&=& \xi[y_l\leftarrow \<q,l>(\varepsilon)\mid l\in [m]]\\
&& \quad\quad [\<q',l>(\varepsilon) \leftarrow t/l[\![s]\!]\mid q'\in Q', t\in \ct_M(s_0,u,q'), l\in [m]]\\
&=& \xi[y_l\leftarrow t/l[\![s]\!]\mid t\in \ct_M(s_0,u,q), l\in [m]]\\
&=& \xi[s_0,u,q]_M. 
\end{eqnarray*}
Since $\theta$ is an empty substitution when $k=0$, $\xi[s_0,u,q]_M=\xi\theta[s_0,u,q]_M$. 

\smallskip\noindent
\emph{Induction step.}
Let $\sigma'=[\sigma,Q']$. 
We denote by $\preceq$ the traverse order on the nodes of $\rhs_M(q,\sigma)$ for all $q\in Q'$ 
when defining $t^{\sigma'}_{i,q,j}$ in the construction. 
Let $c(i, q)$ be the state $q''$ such that $t^{\sigma'}_{i,q,j}$ (for all $j\in [\rank_Q(q)]$) is picked up from 
$\rhs_M(q'',\sigma)$. 
We prove this part by induction on the total order $\preceq$. 
Let $q\in Q'$ and $v\in V(\zeta_q)$, and let $\xi=\zeta_q/v$. 
Let us denote by (IH2) the induction hypothesis of the outer induction on $s$, by (IH$\preceq$) 
that of the inner induction on $\preceq$, and by (IH1) the fact that Statement 1 holds for $s$, 
implied by (IH2) and $2\Longrightarrow 1$.  

\noindent
\emph{Case 1.} $\xi=y_j$. It is trivial from the fact that $\Top'_q(y_j)=\<q, j>(\varepsilon)$. 

\noindent
\emph{Case 2.} $\xi=\delta(\xi_1,\ldots,\xi_l)$.
We get $\xi\theta[s_0,u,q]_M=\delta(\xi'_1, \ldots, \xi'_l)$ where $\xi'_i=$
$\xi_i\theta[s_0,u,q]_M$, and
$\nf(\Rightarrow_{A,s'}, \Top'_{q}(\xi)[\pi\leftarrow \varepsilon])[s_0,u]_A=\delta(\xi''_1,\ldots,\xi''_l)$
where $\xi''_i=\nf(\Rightarrow_{A,s'}, \Top'_{q}(\xi_i)[\pi\leftarrow \varepsilon])[s_0,u]_A$. 
Since $\xi'_i=\xi''_i$ by (IH$\preceq$), this case holds. 

\noindent
\emph{Case 3.} $\xi=\<q', x_i>(\xi_1,\ldots,\xi_m)$. 
\begin{align*}
\xi\theta[s_0,u,q]_M
&=M_{q'}(s/i)[s_0,ui,q']_M\\
&=\nf(\Rightarrow_{A,s'/i}, q'(\varepsilon))[s_0,ui]_A & \text{(by IH1)}.
\end{align*}
On the other hand,
\begin{align*}
&[s_0,ui]_A\\
&=[\<q', l>(\varepsilon) \leftarrow t/l[\![s/i]\!]\mid q'\in Q_i, t\in \ct_M(s_0,ui,q'), l\in [\rank_Q(q')]]\\
&=[\<q', l>(\varepsilon) \leftarrow t^{\sigma'}_{i,q',l}\theta[s_0,u,q'']_M
\mid q'\in Q_i, c(i,q')=q'', l\in [\rank_Q(q')]]\\
&=[\<q', l>(\varepsilon) \leftarrow \nf(\Rightarrow_{A,s'}, \Top'_{q''}(t^{\sigma'}_{i,q',l})[\pi\leftarrow \varepsilon])[s_0,u]_A\\
&\phantom{=[\<q', l>(\varepsilon) \leftarrow t^{\sigma'}_{i,q',l}\theta}
\mid q'\in Q_i, c(i,q')=q'', l\in [\rank_Q(q')]] \quad\quad \text{(by IH2)}\\
&=[\<q', l>(\varepsilon) \leftarrow \nf(\Rightarrow_{A,s'}, \<q',l>(\pi i)[\pi\leftarrow \varepsilon])
\mid q'\in Q_i, l\in [\rank_Q(q')]][s_0,u]_A. 
\end{align*}
Hence, 
\begin{align*}
&\nf(\Rightarrow_{A,s'/i}, q'(\varepsilon))[s_0,ui]_A\\
&=\nf(\Rightarrow_{A,s'/i}, q'(\varepsilon))\\
&\quad\quad\quad[\<q', l>(\varepsilon) \leftarrow \nf(\Rightarrow_{A,s'}, \<q',l>(\pi i)[\pi\leftarrow \varepsilon])\mid q'\in Q_i, l\in [\rank_Q(q')]][s_0,u]_A\\
&=\nf(\Rightarrow_{A,s'}, q'(\pi i)[\pi\leftarrow \varepsilon])[s_0,u]_A\\
&=\nf(\Rightarrow_{A,s'}, \Top'_q(\xi)[\pi\leftarrow \varepsilon])[s_0,u]_A.
\end{align*}

When $s_0=s$ and $u=\varepsilon$, $M_{q_0}(s)[s,\varepsilon,q_0]_M=\nf(\Rightarrow_{A,s'}, q_0(\varepsilon))[s,\varepsilon]_A$. 
Since $\rank_Q(q_0)=0$ and thus $[s,\varepsilon,q_0]_M$ and $[s,\varepsilon]_A$ are empty substitutions, 
$M_{q_0}(s)=\nf(\Rightarrow_{A,s'}, q_0(\varepsilon))$.  
\QED
\end{proof}

Next we show that $\CMTTDFV$ is closed under pre-composition with $\CTREL$. 

\begin{lemma}\label{lem:DFV-TREL-closed}
\rm
For a composition of a $\TREL$ $E$, an $\MTT$ $M$, and a set $L$ of trees,
if $M$ has the dynamic FV property on $E(L)$, then 
there exists an $\MTT$ $M'$ equivalent with $E\com M$ such that $M'$ has 
the dynamic FV property on $L$. 
\end{lemma}
\begin{proof}
Let $L\subseteq T_\Sigma$, and let $E=(Q_E, \Sigma, \Gamma, r_0, R_E)$ and $M=(Q, \Gamma, \Delta, q_0, R)$ 
be an $\TREL$ and an $\MTT$ such that $M$ has the dynamic FV property on $E(L)$. 
We can get an $\MTT$ equivalent with $E\com M$ by the product construction. 
Let $M'=(Q', \Sigma, \Delta, (r_0, q_0), R')$ where $Q'=Q_E\times Q$ 
be the $\MTT$ obtained by the product construction.  
Note that $\rank_{M'}((r,q))=\rank_M(q)$ for $(r,q)\in Q'$. 
Here we show only that $M'$ has the dynamic FV property on $L$. 

Let $s\in L$ and $u\in V(s)$. 
Let $\xi=M'(s[u\leftarrow x])$ and $\<r, x>=E(s[u\leftarrow x])[u]$. 
We have the following properties from the product construction. 
\begin{enumerate}
\item $\xi=M(E(s)[u\leftarrow x])[\![\<q, x>\leftarrow \<(r, q), x>\mid q\in Q]\!]$.
\item $M'_{(r,q)}(s/u)=M_q(E_r(s/u))$.
\end{enumerate}
 
Let $v_1, v_2\in V_{\<Q', \{x\}>}(\xi)$ such that $\xi[v_1]=\xi[v_2]$. 
Then, it follows from the first property that $\xi[v_1]=\xi[v_2]=\<(r, q),x>$ for some $q\in Q$, 
and that $M(E(s)[u\leftarrow x])[v_1]=M(E(s)[u\leftarrow x])[v_2]=\<q,x>$. 
By the dynamic FV property of $M$ on $E(L)$, 
$(M(E(s)[u\leftarrow x])/v_1j)[\![\<q', x> \leftarrow M_{q'}(E(s)/u)\mid q'\in Q]\!]=(M(E(s)[u\leftarrow x])/v_2j)[\![\<q', x>\leftarrow M_{q'}(E(s)/u)\mid q'\in Q]\!]$ 
for every $j\in [\rank_M(q)]$. 
It follows from the above properties that  
$(\xi/v_1j)[\![\<(r,q'), x> \leftarrow M'_{(r,q')}(s/u)\mid q'\in Q]\!]=(\xi/v_2j)[\![\<(r,q'), x>\leftarrow M_{(r,q')}(s/u)\mid q'\in Q]\!]$ for every $j\in [\rank_M((r,q))]$. 
Therefore, $M'$ has the dynamic FV property on $L$.
\QED
\end{proof}

Theorem~\ref{thm:ATTR=MTTRF} and Lemmas~\ref{lem:FVtoDFV}, \ref{lem:MTTDFV-TREL+ATT}, and \ref{lem:DFV-TREL-closed} 
yield the main result of this section. 
\begin{theorem}\label{thm:MTTRDFV=MTTUDFV}
\rm
$\CATTU=\CMTTRDFV = \CMTTUDFV$ 
\end{theorem}

Let us consider the decidability of the dynamic FV property. 
While the FV property is easily decidable, 
we do not know how to decide the dynamic FV property; in fact, we are able to show
that this problem is at least as difficult as deciding  equivalence  of ATTs. 
The proof constructs from two given ATTs an $\MTTR$ which has the dynamic FV property
if and only if the ATTs are equivalent. This is done by nesting the start calls
for the ATTs under a new fixed state. 
\begin{theorem}\label{l:equiv}
\rm
Deciding the dynamic FV property for $\CMTTR$ is at least as hard
as deciding equivalence of ATTs.
\end{theorem}
\begin{proof}
We show that the equivalence problem of ATTs can be reduced to the decision problem of the dynamic FV property for $\CMTTR$. 
Let $A_1$ and $A_2$ be $\ATT$s. 
From Lemmas~\ref{lem:FVo-ourFV1} and \ref{lem:FVtoDFV}, 
for $i=1,2$, we can get a composition of a $\BREL$ $E_i$ 
and an MTT $M_i=(Q_i, \Sigma_i,\Delta,q_{i0},R_i)$ 
such that it is equivalent with $A_i$ and it has the dynamic FV property.  
We can assume with loss of generality that $Q_1\cap Q_2=\emptyset$. 
Let $\Sigma_{M}=\{(\sigma_1,\sigma_2)^{(k)}\mid \sigma_1\in \Sigma_1^{(k)}, \sigma_2\in \Sigma_2^{(k)}\}$. 
We assume that $\Sigma^{(1)}_{M}\neq \emptyset$, $e\in \Delta^{(0)}$, and $\delta\in \Delta^{(2)}$. 

Let $E$ be the relabeling such that for every input tree $t$ $E(t)$ is the convolution tree of $E(t_1)$ 
and $E_2(t)$, i.e., $E(t)[u]=(E_1(t)[u], E_2(t)[u])$ for every $u\in V(t)$. 
Let $M=(Q_1\cup Q_2\cup \{q_0^{(0)}, {q'}^{(1)}\}, \Sigma_M, \Delta, q, R_1'\cup R_2'\cup R)$ 
where $q, q'$ are new distinct states not in $Q_1\cup Q_2$ such that for every $a\in \Sigma_M^{(1)}$ 
let the rule
\[
\<q_0, a(x_1)>\to \delta(\<q',x_1>(\<q_{10}, x_1>), \<q',x_1>(\<q_{20}, x_1>))
\]
be in $R$.  
Let the rule $\<q_0, \sigma(x_1,\ldots,x_k)>\to e$ be in $R$ 
for every $\sigma\in \Sigma^{(k)}_M$ with $k\neq 1$, 
and let the rule $\<q', \sigma(x_1,\ldots,x_k)>(y_1)\to y_1$ be in $R$ for every $\sigma\in \Sigma^{(k)}$ with $k\geq 0$. 
For $q\in Q_i$, let the rule $\<q, [\sigma_1,\sigma_2](x_1,\ldots, x_k)> \to \rhs_{M_i}(q, \sigma_i)$ be in $R'_i$.
Since $E_1\com M_1$ and $E_2\com M_2$ have the dynamic FV property and $Q_1\cap Q_2=\emptyset$, 
it follows that $M$ has the dynamic FV property if and only if 
$M_1(E_1(s))=M_{q_{10}}(E(s))=M_{q_{20}}(E(s))=M_2(E_2(s))$ for every $s\in T_\Sigma$. 
\QED
\end{proof}

\section{Conclusions}

We have presented two new characterizations of attributed tree transformation
with regular look-around 
in terms of macro tree transducers
(with regular look-around or look-ahead): first a static restriction that is
similar as the one given by F{\"u}l{\"o}p and Vogler~\cite{FV99},
but is extended to nondeleting MTTs with regular look-ahead.
We show that for every MTT with the restriction, an equivalent
\emph{non-circular} ATT can be constructed. 
Our second restriction (called \emph{dynamic FV property})
requires that during any computation, all $j$-th parameter trees of a given state
evaluate to the same output tree. This restriction captures many more MTTs than
previous restrictions, however, it remains an open problem how to decide this
restriction.

One may wonder if every MTT that has the {\bf LIN} property (cf. the Introduction)
also has the dynamic FV property. Alas, this is not the case:
\[
\begin{array}{lcllcl}
\la q_0,\#(x_1)\ra&\to& \la q,x_1\ra(e)\\
\la q,a(x_1)\ra(y_1)&\to& f(f(\la q,x_1\ra(y_1),\la q,x_1\ra(y_1)),\la q,x\ra(f(y_1,y_1)))\\
\la q,e\ra(y_1)&\to& y_1
\end{array}
\]
After $n$ applications of the second rule, state $q$ has $n+1$ distinct parameter trees.
Hence the dynamic FV property is violated.
Note that this MTT translates monadic trees into full binary
trees (this can be done by a simple top-down or bottom-up transducer).
We would like to know, if there is a normal form that guarantees
the dynamic FV property for every MTT with the {\bf LIN} property.



\bibliographystyle{elsarticle-num} 
\bibliography{main}

%
%
%
%
\end{document}